\documentclass[12pt]{article}
\usepackage[margin=1in,letterpaper]{geometry}

\usepackage{amsmath,amsthm}  
\usepackage{newtxtext,newtxmath}
\usepackage{setspace}
\usepackage{float}
\usepackage{microtype}  

\usepackage{natbib}
\usepackage{graphicx}
\usepackage{subcaption}
\usepackage{booktabs}  

\usepackage{enumitem}
\setlist{nosep}  

\usepackage{xcolor}
\usepackage{hyperref}
\hypersetup{
    colorlinks=true,
    linkcolor=black,         
    citecolor=blue!50!black, 
    urlcolor=blue!50!black
}

\theoremstyle{plain}
\newtheorem{theorem}{Theorem}
\newtheorem{lemma}{Lemma}
\newtheorem{proposition}{Proposition}
\newtheorem{corollary}{Corollary}
\newtheorem{definition}{Definition}
\newtheorem{axiom}{Axiom}
\newtheorem{example}{Example}

\theoremstyle{remark}
\newtheorem{remark}{Remark}

\newcommand{\X}{\mathcal{X}}
\newcommand{\blkdiag}{\operatorname{blkdiag}}
\newcommand{\diag}{\operatorname{diag}}

\usepackage{fancyhdr}
\renewcommand{\headrulewidth}{0pt}
\begin{document}

\title{Aspiration-Weighted Influence}
\author{%
  {\large Siming Ye}\footnote{%
    Department of Economics, \emph{Georgetown University}. 
    Email: \texttt{sy677@georgetown.edu}. 
    I am deeply grateful to my advisor, Christopher Chambers, for his guidance and support. I thank Yusufcan Masatlioglu, Federico Echenique, John Rust, Garance Genicot, Christopher Turansick, Axel Anderson, Frank Yang, Luca Anderlini, Zhengxing Huang, Tri Phu Vu, Emilio Muniz Langle, and seminar participants at Shanghai University of Finance and Economics, for their helpful comments. Earlier versions of this paper were presented at the 2025 Midwest Economics Association Annual Meeting and a 2025 PhD workshop at Queen Mary University of London. I thank the participants at these events for their valuable feedback. All errors are my own.
  }%
}
\date{\today}
\maketitle
\begin{abstract} 
We study directed social influence when an influencer chooses from a richer menu than a constrained follower (decision maker). The decision maker (DM) selects from a feasible set, while the influencer displays a distribution over a superset that includes infeasible alternatives. We propose the Aspiration-Weighted Luce Model (AWLM): the DM forms a convex combination of her idiosyncratic Luce preferences within the feasible set and the influencer's distribution, then renormalizes this attempt target onto the feasible set. This renormalization generates an \emph{aspirational dampening} effect: holding the influencer's within-feasible composition fixed and shifting exposure toward infeasible alternatives attenuates influence on feasible choices. We provide an axiomatic characterization based on (i) proportional responses to shifts in feasible exposure and (ii) a unit-slope leverage restriction across different levels of feasible share. The model allows for point identification of influence strength and idiosyncratic preferences from two exposure regimes, yielding testable overidentifying restrictions for empirical application. 
\end{abstract}
\pagestyle{fancy}
\fancyhf{} 
\renewcommand{\headrulewidth}{0pt} 
\fancyfoot[C]{\thepage}
\pagestyle{fancy} 
\thispagestyle{fancy}

\section{Introduction}
How does social influence operate when the influencer chooses from a richer menu than the follower? Celebrities, experts, and wealthier peers routinely display behavior over alternatives that are infeasible for aspiring observers. An amateur photographer follows professionals whose equipment exceeds her budget; a home cook watches chefs use ingredients unavailable locally; a teenager tracks a blogger's high-street fashion choices she cannot replicate. In each case, the follower ultimately chooses from a feasible set $S$, while the observed behavior is generated on a strict superset $I \supset S$. This paper asks how exposure defined partly on $I\setminus S$ reshapes choice over $S$.

The question bears several levels of economic meaning. Aspirational marketing is a deliberate strategy: brands cultivate images centered on products that most consumers cannot afford, hoping this will influence purchases of entry-level items. Whether this strategy succeeds depends on how exposure to unattainable goods transmits to feasible consumption. The rapid expansion of social media has dramatically increased exposure to consumption patterns of individuals with different opportunity sets. Understanding how this exposure operates is also essential for analyzing modern consumer behavior. Additionally, from a revealed-preference perspective, influence that depends on infeasible alternatives elicits a conceptual challenge: the standard assumption is that choices reveal preferences over the available set, but aspirational influence means that behavior outside this set also matters. 

Theoretically, this setting introduces a new source of variation to the decision-theoretic apparatus. In standard models (e.g., Random Utility, Luce), choice probabilities change only when the \emph{feasible set} changes or when the underlying preferences shift structurally. Here, we identify a channel where choice probabilities vary due to changes in the \emph{environment} (the influencer's behavior on the infeasible set), while the decision maker's feasible set and idiosyncratic preferences remain fixed. This allows us to exploit \text{variation in the unattainable} to identify structural parameters governing the attainable. This is distinct from standard peer effects, where the influencer and follower are typically assumed to face identical choice sets.

We address these questions by developing the Aspiration-Weighted Luce Model (AWLM) within a stochastic choice framework, capturing influence as a shift in consumption propensities rather than deterministic constraints. The DM has idiosyncratic preferences satisfying Luce's choice axiom, where choice probabilities are proportional to underlying utilities. Under exposure to an influencer whose demonstrated behavior spans a larger menu, these idiosyncratic probabilities are mixed with the influencer's distribution and then renormalized over the DM's feasible set. The mixing weight captures the DM's susceptibility to influence; the renormalization captures the constraint that she must ultimately choose from what is available. This mixing-then-renormalizing structure is the core of the model, and the renormalization step is where the economic content resides.

To understand the mechanism, it is useful to decompose the influencer's behavior into two components: her \emph{composition} over feasible alternatives, $q(\cdot\mid S)$, and her \emph{feasible share}, $q_S$ (the mass devoted to alternatives available to the DM). In the AWLM, changes in the influencer's feasible composition translate into proportional shifts in the DM’s choices. Crucially, however, the strength of this response is governed by the feasible share $q_S$.

This dependency generates the model’s main comparative-static prediction: \emph{aspirational dampening}. Holding fixed the influencer’s composition over feasible alternatives, shifting exposure toward infeasible goods attenuates influence on feasible choices. Intuitively, aspirational content is ``wasted'' from the perspective of influencing the DM's feasible consumption: it increases the chance that an influencer-driven choice that lands outside the feasible set and would be discarded. An influencer who showcases products the DM can actually afford exerts stronger leverage than one who mostly features luxury items, even if their recommendations among affordable options are identical. Thus, building an aspirational image by featuring exclusive products comes at the direct cost of diluting influence over budget-constrained consumers.

Our results are organized into three parts. We begin by providing a behavioral characterization of the AWLM based on two restrictions: \emph{controlled collinearity}, which governs response to composition shifts, and an \emph{affine aspiration penalty}, which governs the dampening effect. Next, we establish that the model's parameters (influence strength and idiosyncratic preferences) are nonparametrically identified from as few as two distinct exposure conditions. Notably, this identification strategy does not rely on variation in the DM's feasible set. Instead, it exploits the \emph{geometry of exposure}: observed choice shifts must be collinear with exposure shifts in a specific way that uniquely pins down the parameters. This provides a toolkit for recovering preferences even in environments where the decision maker faces static constraints (e.g., a fixed budget), provided there is sufficient variation in the aspirational environment. Finally, we develop consistent estimators for settings with noisy choice data.

The remainder of the paper is organized as follows. Section~\ref{sec:related_literature} situates our contribution at the intersection of stochastic choice and social influence, identifying the theoretical gap regarding influence from infeasible alternatives. Section~\ref{sec:model} formalizes the model and provides a microfoundations. We also analyze the mechanism's comparative statics to uncover the \emph{aspirational dampening} effect. Section~\ref{sec:characterization} provides the axiomatic foundation. Finally, Section~\ref{sec:identification} bridges the model to data: we derive a constructive identification strategy to recover the unobservable influence parameters and idiosyncratic preferences from limited exposure variation.

\section{Related Literature}
\label{sec:related_literature}

\paragraph{Stochastic Choice and Social Influence}
Our framework builds on the foundation of the Luce model \citep{Luce_1959} and equivalently the multinomial logit model \citep{McFadden1974}. Recent theoretical work integrates peer influence: \citet{Chambers_Cuhadaroglu_Masatlioglu_2022} model cardinal peer influence, while \citet{Tugce2017} examines mutual influence in interactive decision-making using a choice-theoretic lens. 

A parallel literature relaxes the assumption of fixed attention. \citet{Masatlioglu_Nakajima_Ozbay_2012} and \citet{Manzini_Mariotti_2014} demonstrate how limited consideration distorts choice; \citet{Kops2019Social} and \citet{KashaevLazzatiXiao2023} further explore how social environments modify these consideration sets. Other work departs from standard Luce through different channels: \citet{MatejkaMcKay2015} show that rational inattention generates a modified logit where prior beliefs enter as
action-specific weights, while \citet{ChambersMasatliogluTuransick2024} characterize correlated stochastic choice, relaxing the cross-menu independence implicit in standard random utility models.  AWLM introduces dependence through yet another mechanism: the influencer's distribution enters choice probabilities directly via mixing, and the influencer and decision maker may face different menus.

\paragraph{Peer Effects and Conformity}
Our identification strategy engages with the identification of peer effects, a literature grappling with the reflection problem \citep{Manski_1993, Blume2011, BramoulleEtAl2020}. Theoretical mechanisms in this domain range from status-driven conformity \citep{Bernheim_1994} and social utility alignment \citep{Brock_Durlauf_2001} to network-mediated interactions \citep{LinearInteractions2015, NetworkNorm2020, Boucher2024, Volpe2023}. While empirical strategies have successfully disentangled these effects using random assignment or field experiments \citep{Sacerdote2001, BursztynEtAl2014}, existing models largely focus on symmetric or network influence. In contrast, the AWLM focuses on \emph{directed} influence from a specific figure whose menu strictly contains the decision maker's feasible set.

\paragraph{Aspirations and Reference-Dependent Choice}
Finally, we formalize the role of aspirations. Economics has long recognized the importance of the ``capacity to aspire'' \citep{Appadurai2004} and how reference points shape incentives \citep{KoszegiRabin2006, Genicot2017Aspiration, Dalton_Ghosal_Mani_2016} and status competition \citep{Veblen_1899, Bagwell_Bernheim_1996, Hopkins_Kornienko_2004, Mookherjee_Napel_Ray_2010}. Most directly related is \citet{Guney_Richter_Tsur_2018}, who study deterministic choice steered by unavailable alternatives. We extend this insight to a stochastic framework. Specifically, our model captures ``aspirational dampening,'' where exposure to infeasible alternatives attenuates influence. This is distinct from a complementary ``halo effect,'' which would require a different modeling approach. By focusing on the observed distribution of the aspirational figure, the AWLM generates clean, testable restrictions on choice probabilities.

\section{The Model}
\label{sec:model}
Let $\X$ be a finite universe of alternatives. In standard stochastic choice theory, a \emph{stochastic choice rule} assigns to each nonempty feasible set $S\subseteq \X$ a probability distribution over $S$. We distinguish the follower's feasible set $S \subseteq \X$ from the influencer's menu $I \subseteq \X$, with $S \subseteq I$. The influencer's behavior is represented by an exposure distribution $q$ supported on $I$. Throughout, we identify $q\in\Delta(I)$ with its zero-extension to $\X$, so exposures can be treated as elements of $\Delta(\X)$ without additional notation.

For any exposure $q$ and feasible set $S$, define the \emph{feasible share}
\[
q_S \equiv \sum_{x\in S} q(x)\in[0,1],
\]
and call $1-q_S$ the \emph{aspirational mass}. When $q_S>0$, write $q(\cdot\mid S)\equiv q|_S/q_S\in\Delta(S)$
for the influencer's within-feasible composition.

\subsection{Idiosyncratic preferences: Luce}
The DM has idiosyncratic Luce weights $u:\X\to\mathbb{R}_{++}$. Her idiosyncratic (uninfluenced) choice rule is
\[
p_0(x\mid S)=\frac{u(x)}{\sum_{y\in S}u(y)},
\qquad x\in S.
\]

\subsection{Aspiration-weighted influence}

Fix an influence strength $\alpha\in[0,1]$. Given $(S,q)$, the induced feasible choice rule is
\begin{equation}
\label{eq:awlm-closed}
p(x\mid S;q,\alpha)
=
\frac{(1-\alpha)p_0(x\mid S)+\alpha q(x)}{(1-\alpha)+\alpha q_S},
\qquad x\in S.
\end{equation}
The numerator mixes the DM's idiosyncratic propensity on $S$ with the influencer's raw exposure on $I$.
The denominator imposes feasibility: the realized choice must lie in $S$.

Equivalently, AWLM is ``mix then normalize.'' Let $\rho^S$ be the zero-extension of $p_0(\cdot\mid S)$ to $\X$,
\[
\rho^S(x)\equiv p_0(x\mid S)\,\mathbf{1}\{x\in S\},
\qquad x\in\X,
\]
and define the \emph{attempt target} on $\X$ by
\[
M(\cdot\mid S;q)\equiv (1-\alpha)\rho^S+\alpha q.
\]
Then $p(\cdot\mid S;q,\alpha)$ is exactly $M(\cdot\mid S;q)$ renormalized onto $S$.

\begin{definition}[Aspiration-Weighted Luce Model (AWLM)]
An exposure-dependent stochastic choice rule $(S,q)\mapsto p(\cdot\mid S;q)$ follows an
\emph{Aspiration-Weighted Luce Model} if there exist a Luce weight function
$u:\X\to\mathbb{R}_{++}$ and an influence parameter $\alpha\in(0,1)$\footnote{The cases $\alpha=0$ (no influence) and $\alpha=1$ (pure imitation) are obtained as limits. When $\alpha=1$, the induced choice is $p(\cdot\mid S;q,1)=q(\cdot\mid S)$ provided $q_S>0$.} such that for every
pair of menus $S\subseteq I\subseteq \X$ and every exposure $q\in\Delta(I)$, the induced
choice probabilities satisfy \eqref{eq:awlm-closed}.
\end{definition}

\paragraph{Figure conventions.}
In all simplex diagrams, red points denote influencer exposures $q\in\Delta(\mathcal X)$, blue points denote the DM's idiosyncratic distribution $p_0(\cdot\mid S)\in\Delta(S)$ (embedded in $\Delta(\mathcal X)$ by zero-extension), gray points denote the attempt-level target $M=(1-\alpha)p_0+\alpha q$, green points denote the realized feasible choice $p(\cdot\mid S;q,\alpha)$ obtained by normalizing $M$ onto $S$, and orange points denote the conditional feasible exposure $q(\cdot\mid S)$.

\begin{figure}[H]
  \centering
  \includegraphics[width=0.75\textwidth]{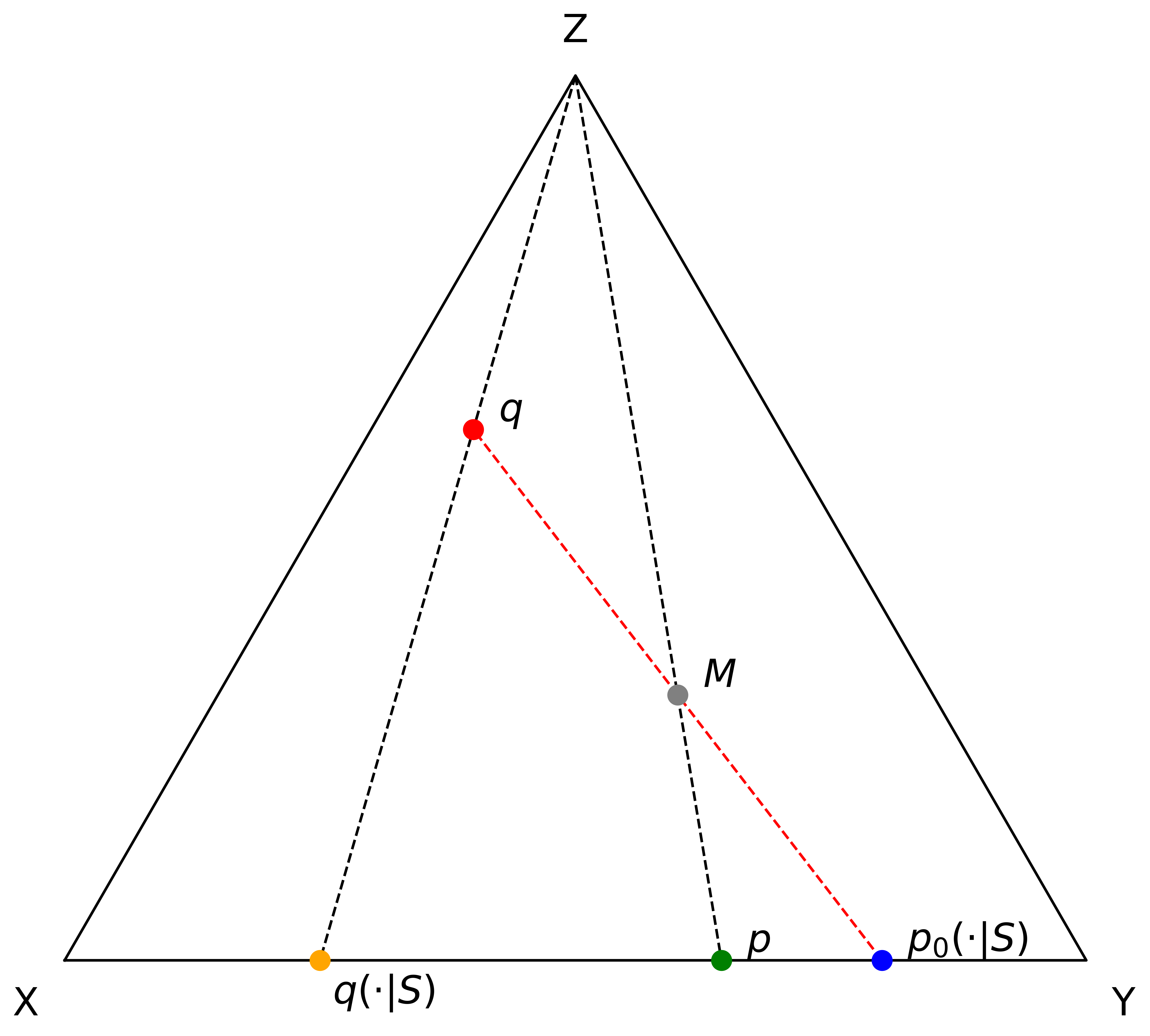}
  \caption{AWLM as ``mix then normalize.'' The influencer's exposure $q$ is mixed with the idiosyncratic feasible distribution $p_0(\cdot\mid S)$ to form the attempt-level target $M=(1-\alpha)p_0+\alpha q$ on the chord $[p_0,q]$. The realized choice $p(\cdot\mid S;q,\alpha)$ is the normalization of $M$ onto the feasible set $S$, represented by the dotted conditioning ray through $M$. The point $q(\cdot\mid S)$ is the analogous projection of $q$ onto $S$.}
  \label{fig:awlm-overview}
\end{figure}

\paragraph{Special case: influence within the feasible set.}
If $S=I$ (no aspirational alternatives), then $q_S=1$ and the normalization disappears:
\[
p(\cdot\mid I;q,\alpha)=(1-\alpha)p_0(\cdot\mid I)+\alpha q(\cdot),
\]
which is the standard convex-mixture form of within-menu influence, as shown in Figure~\ref{fig:benchmark}.

\begin{figure}[H]
  \centering
  \includegraphics[width=0.5\textwidth]{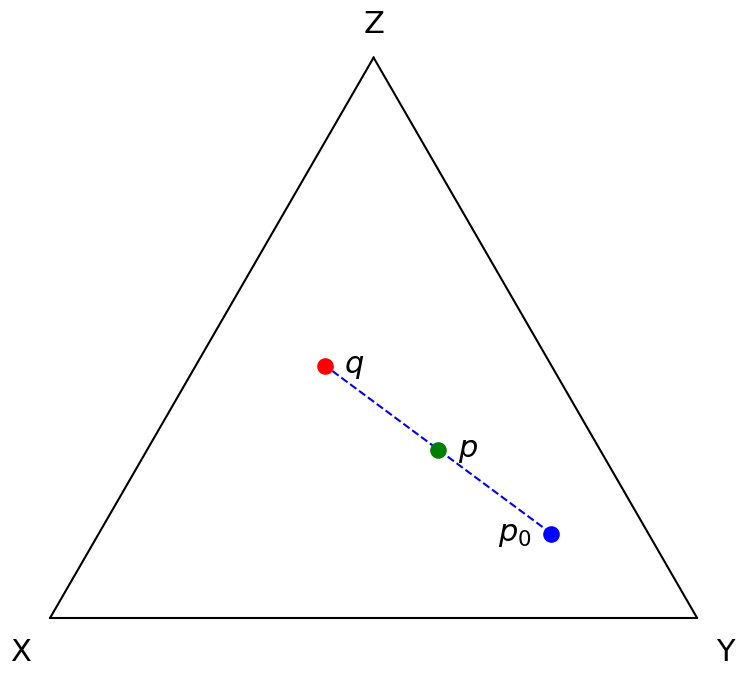}
  \caption{Benchmark: influence within the feasible set ($S=I$). When the influencer's menu coincides with the DM's feasible set, $q_S=1$ and no normalization is needed. The DM's influenced choice $p$ lies on the chord joining $p_0$ and $q$ at position $p = (1-\alpha)p_0 + \alpha q$. This is the standard convex-mixture form of influence \citep{Chambers_Cuhadaroglu_Masatlioglu_2022}.}
  \label{fig:benchmark}
\end{figure}

The AWLM formula~\eqref{eq:awlm-closed} takes as input the influencer's full distribution $q \in \Delta(I)$, which may assign positive mass to alternatives outside the feasible set~$S$. This infeasible portion of the exposure matters only through the total mass $q_S$ diverted away from~$S$, not through how that mass is distributed among infeasible alternatives.
\begin{lemma}[Intra-aspiration irrelevance]
Fix $S\subseteq I$ and $\alpha\in[0,1]$. If $q,q'\in\Delta(I)$ satisfy $q|_S=q'|_S$, then $p(\cdot\mid S;q,\alpha)=p(\cdot\mid S;q',\alpha)$.
\label{lem:outside-irrelevance}
\end{lemma}

\begin{proof}
Immediate from \eqref{eq:awlm-closed}: for $x\in S$, the choice probability depends on $q$ only through $\{q(x)\}_{x\in S}$ and $q_S$.  A visualization can be found in Appendix~\ref{app:figures} Figure~\ref{fig:intra-aspiration}
\end{proof}
Lemma~\ref{lem:outside-irrelevance} implies that reallocations of exposure among infeasible alternatives, however dramatic, have no effect on choice over~$S$.

\subsection{Microfoundation: sampling until feasible}

We now provide a procedural foundation for the AWLM.
Recall the zero-extension $\rho^S$ and the attempt target $M(\cdot\mid S;q)$ defined above.

Consider the following procedure. Draw i.i.d.\ samples $\{Y_t\}_{t\ge 1}$ from the attempt target
$M(\cdot\mid S;q)=(1-\alpha)\rho^S+\alpha q$. Let
$\tau\equiv \inf\{t\ge 1:\,Y_t\in S\}, X\equiv Y_\tau.$ Assume $M(S\mid S;q)>0$ (equivalently $(1-\alpha)+\alpha q_S>0$), so $\tau<\infty$ almost surely

\begin{proposition}[Sampling-until-feasible implies AWLM]\label{prop:micro_awlm}
The procedure induces, for each exposure $q$, a choice distribution on $S$ given by
\[
p(x\mid S;q)
=\frac{M(x\mid S;q)}{M(S\mid S;q)}
=\frac{(1-\alpha)p_0(x\mid S)+\alpha q(x)}{(1-\alpha)+\alpha q_S},
\qquad x\in S.
\]
\end{proposition}

\begin{proof}
Fix $x\in S$. Since attempts are i.i.d.\ and $\tau$ is the first time a draw lands in $S$,
\[
p(x\mid S;q)
=\sum_{t=1}^{\infty}\Pr(\tau=t,\,Y_t=x)
=\sum_{t=1}^{\infty}\Pr(Y_1\notin S)^{t-1}\Pr(Y_1=x).
\]
Because $Y_1\sim M(\cdot\mid S;q)$, we have $\Pr(Y_1\notin S)=1-M(S\mid S;q)$ and $\Pr(Y_1=x)=M(x\mid S;q)$.
Thus
\[
p(x\mid S;q)
=\sum_{t=1}^{\infty}\bigl(1-M(S\mid S;q)\bigr)^{t-1}M(x\mid S;q)
=\frac{M(x\mid S;q)}{M(S\mid S;q)}.
\]
The second equality follows from substituting the definitions of $M(\cdot\mid S;q)$ and $M(S\mid S;q)$.
\end{proof}

The microfoundation\footnote{This procedure has a natural random utility interpretation as a two-class latent logit. At each attempt, the DM enters ``aspiration influence mode'' with probability $\alpha$ (drawing from $q$ via logit with utilities $\log q(x)$) or ``idiosyncratic mode'' with probability $1-\alpha$ (drawing from $p_0(\cdot\mid S)$ via logit with utilities $\log u(x)$). The acceptance filter then conditions on feasibility.} clarifies the role of the two influence parameters that will recur throughout this paper. The primitive $\alpha$ governs influence \emph{per-attempt}: it is the weight placed on the influencer's distribution each time the DM samples from her attempt target $M$. The outcome-level effective weight,
\[
\beta(q_S) \equiv \frac{\alpha q_S}{(1-\alpha)+\alpha q_S},
\]
governs influence among \emph{realized} feasible choices.  Equivalently, writing
$q(\cdot\mid S)\equiv q|_S/q_S$ for the influencer's within-feasible composition,
\begin{equation}
\label{eq:beta-mixture}
p(\cdot\mid S;q,\alpha)=\bigl(1-\beta(q_S)\bigr)p_0(\cdot\mid S)+\beta(q_S)\,q(\cdot\mid S).
\end{equation}

Because $(1-\alpha)+\alpha q_S > 0$, we have $\beta(q_S) \ge 0$; and since
$q_S \le 1$, we have $\beta(q_S) \le \alpha$.  Differentiating shows
$\partial\beta/\partial q_S > 0$ for $\alpha \in (0,1)$. Since attempts landing on infeasible alternatives are discarded and redrawn, $\beta(q_S)\le\alpha$, with equality only when $q_S=1$. This wedge is the source of \emph{aspirational dampening}, developed in Subsection~\ref{subsec:comparative}.

\paragraph{Alternative failure rules.}
If the DM draws once from $M$ and, upon hitting an infeasible alternative, reverts immediately to $p_0(\cdot\mid S)$, the induced choice becomes $(1-\alpha q_S)p_0(x\mid S)+\alpha q(x)$. More generally, if she retries with probability $r\in[0,1]$ after an infeasible draw, the effective weight becomes $\beta_r(q_S)=\alpha q_S/[1-\alpha r(1-q_S)]$. The AWLM corresponds to full persistence $r=1$.

\subsection{Comparative Statics}
\label{subsec:comparative}

The mixture representation \eqref{eq:beta-mixture} decomposes influence into two objects: the influencer's within-feasible composition $q(\cdot\mid S)$ and the effective weight $\beta(q_S)$ placed on it.  We now examine how feasible choice responds to changes in each. Throughout this subsection fix $S\subseteq I$, $p_0(\cdot\mid S)\in\Delta(S)$, and an exposure $q$ with $q_S>0$.

\paragraph{Varying $\alpha$ at fixed exposure.}
Holding the exposure $q$ fixed, increasing $\alpha$ raises $\beta(q_S)$ and moves the induced choice along the segment from $p_0(\cdot\mid S)$ toward $q(\cdot\mid S)$. This is immediate: $\partial\beta/\partial\alpha = q_S/(1-\alpha+\alpha q_S)^2 > 0$. Figure~\ref{fig:alpha-comparison} illustrates.

Figure~\ref{fig:alpha-comparison} visualizes this monotone movement: for fixed $(p_0,q)$, increasing $\alpha$ moves $p(\cdot\mid S;q)$ closer to the influencer.
\begin{figure}[H]
  \centering
  \begin{subfigure}[b]{0.48\textwidth}
    \centering
    \includegraphics[width=\textwidth]{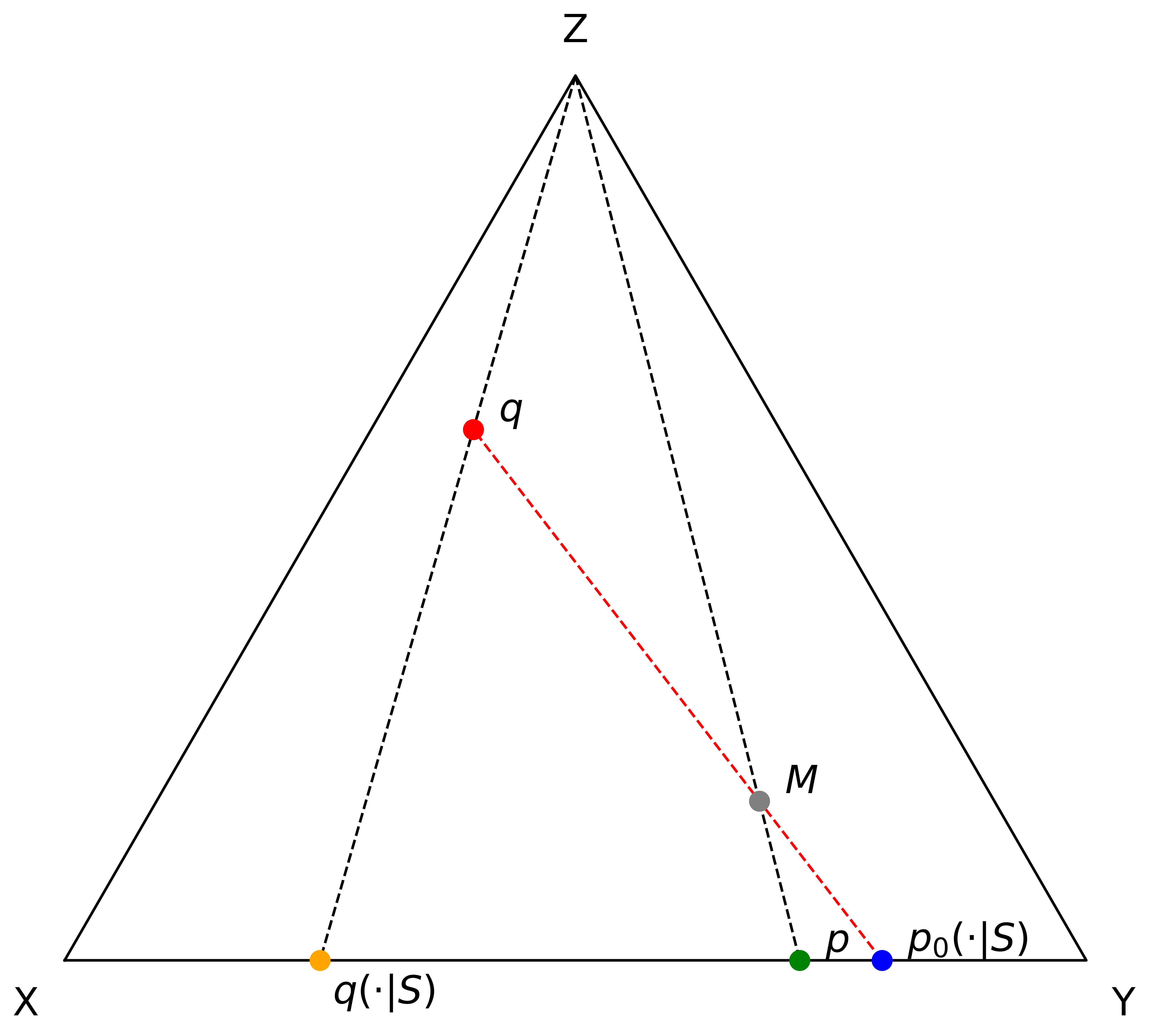}
    \caption{$\alpha = 0.3$}
    \label{fig:alpha-low}
  \end{subfigure}
  \hfill
  \begin{subfigure}[b]{0.48\textwidth}
    \centering
    \includegraphics[width=\textwidth]{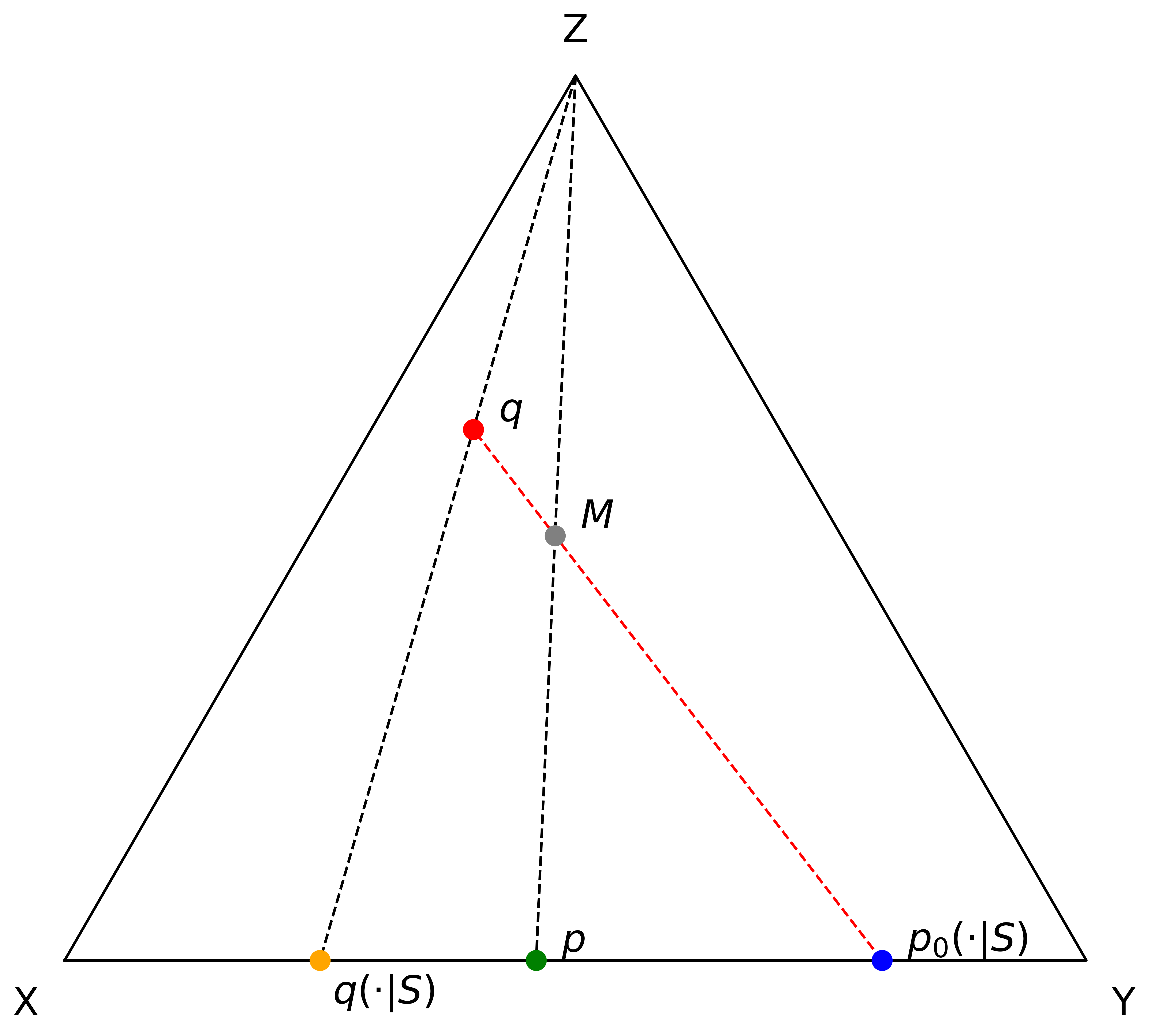}
    \caption{$\alpha = 0.8$}
    \label{fig:alpha-high}
  \end{subfigure}
  \caption{Changing $\alpha$ at fixed exposure. Holding $(S,p_0,q)$ fixed, increasing $\alpha$ shifts the attempt target $M$ toward $q$ and moves the realized feasible choice $p(\cdot\mid S;q,\alpha)$ along the segment joining $p_0(\cdot\mid S)$ and $q(\cdot\mid S)$. The displacement is governed by the effective weight $\beta(q_S)$.}
  \label{fig:alpha-comparison}
\end{figure}

\paragraph{Aspirational dampening as $q_S$ falls}
A more distinctive comparative static concerns infeasible alternatives.
Suppose the influencer reallocates mass from $S$ to $I\setminus S$ while
holding her within-feasible composition $q(\cdot\mid S)$ fixed.

\begin{proposition}[Aspirational dampening]
\label{prop:dampening}
Fix $S$, $p_0(\cdot\mid S)$, and $\alpha\in(0,1)$. Let $q,q'\in\Delta(\mathcal{X})$ satisfy
$q(\cdot\mid S)=q'(\cdot\mid S)$ and $q_S,q_S'>0$. Then
\[
p(\cdot\mid S;q')-p(\cdot\mid S;q)
=
\bigl(\beta(q_S')-\beta(q_S)\bigr)\Big(q(\cdot\mid S)-p_0(\cdot\mid S)\Big),
\]
so the entire effect of changing the infeasible mass operates through the scalar $\beta(q_S)$.
In particular, if $q_S'<q_S$ then $\beta(q_S')<\beta(q_S)$ and
\[
p(\cdot\mid S;q')
\ \text{lies closer to}\ p_0(\cdot\mid S)\ \text{than}\ p(\cdot\mid S;q)
\ \text{along the line segment joining}\ p_0(\cdot\mid S)\ \text{and}\ q(\cdot\mid S).
\]
Equivalently, for every $x\in S$,
\[
p(x\mid S;q)-p_0(x\mid S)=\beta(q_S)\Big(q(x\mid S)-p_0(x\mid S)\Big),
\]
so the deviation from idiosyncratic scales one-for-one with $\beta(q_S)$.
\end{proposition}

Proposition~\ref{prop:dampening} gives a transparent interpretation of the normalization step. When the influencer allocates more probability mass to infeasible options (smaller $q_S$), the DM is less likely to ``hit'' a feasible influencer-drawn option on any given attempt. Under the sampling-until-feasible microfoundation, this reduces the chance that the final realized feasible choice is directly governed by the influencer rather than by the idiosyncratic choice. The wedge between the attempt-levelparameter $\alpha$ and the outcome-level effective weight $\beta(q_S)$ is therefore the model's central aspirational channel.

Figure~\ref{fig:aspiration-dampening} illustrates Proposition~\ref{prop:dampening}.

\begin{figure}[ht!]
  \centering
  \begin{subfigure}[b]{0.48\textwidth}
    \centering
    \includegraphics[width=\textwidth]{highZalpha5.png}
    \caption{High aspirational mass: $q_S = 0.4$}
    \label{fig:high-aspiration}
  \end{subfigure}
  \hfill
  \begin{subfigure}[b]{0.48\textwidth}
    \centering
    \includegraphics[width=\textwidth]{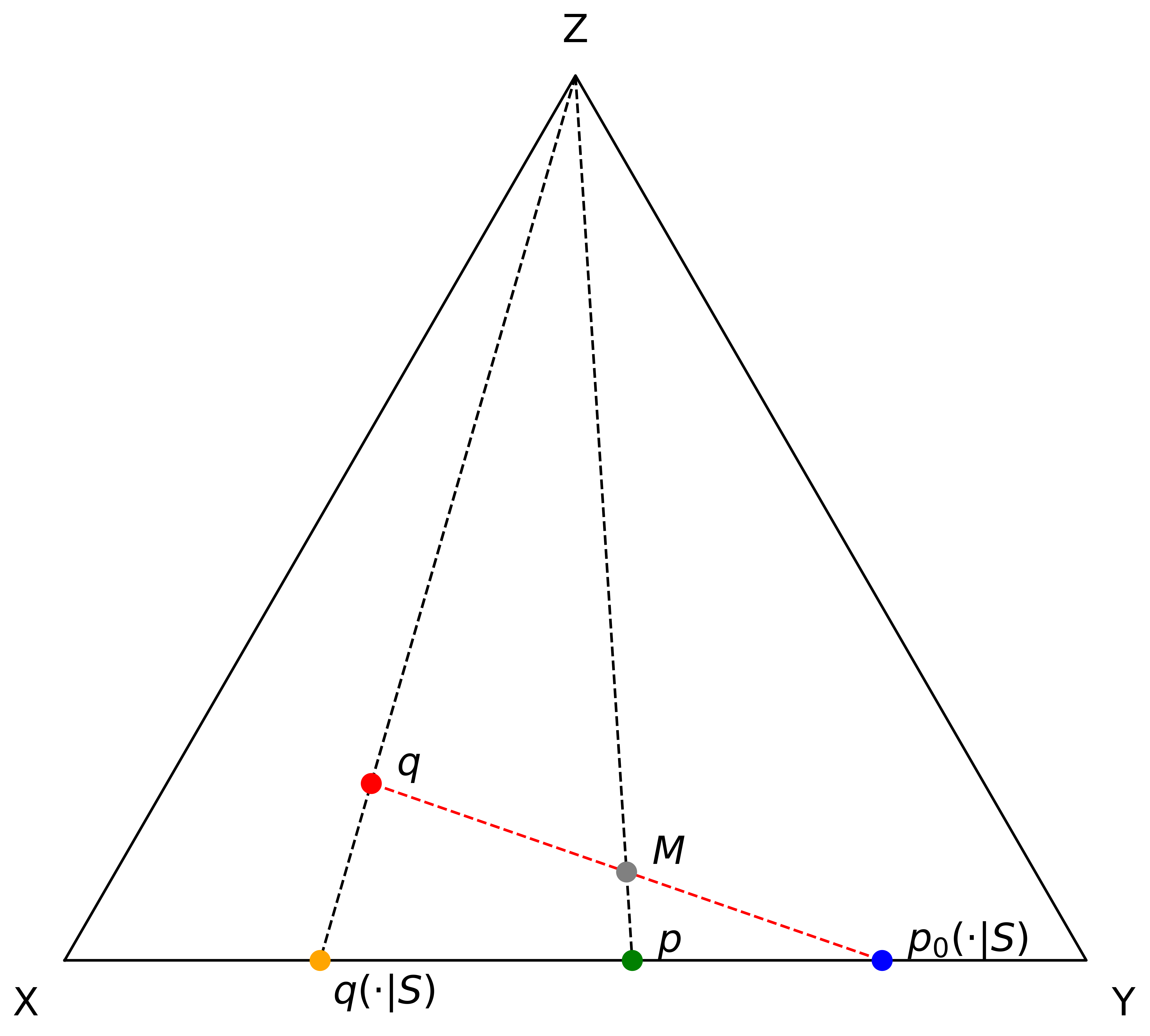}
    \caption{Low aspirational mass: $q_S = 0.8$}
    \label{fig:low-aspiration}
  \end{subfigure}
  \caption{Aspirational dampening (Proposition~\ref{prop:dampening}). The within-$S$ composition $q(\cdot\mid S)$ is held fixed (orange), while the feasible share $q_S$ changes, moving $q$ (red) along the ray toward the infeasible vertex. Since $p(\cdot\mid S;q,\alpha)=(1-\beta(q_S))p_0+\beta(q_S)q(\cdot\mid S)$ with $\beta(\cdot)$ increasing in $q_S$, a lower $q_S$ reduces $\beta(q_S)$ and pulls the realized choice $p$ (green) toward $p_0$ (blue).}
  \label{fig:aspiration-dampening}
\end{figure}

\section{Characterization}
\label{sec:characterization}

This section gives an axiomatic characterization of the exposure-dependent choice rules generated by the Aspiration-Weighted Luce Model (AWLM). We first work menu-by-menu: fixing a feasible set $S$ with $|S|\ge 3$, we characterize when the rule $q\mapsto p(\cdot\mid S;q)$ can be written as
\begin{equation}\label{eq:awlm-permenu}
p(x\mid S;q)=\frac{(1-\alpha)p_0(x)+\alpha q(x)}{(1-\alpha)+\alpha q_S},
\qquad x\in S,
\end{equation}
for some $(\alpha,p_0)\in(0,1)\times \Delta(S)$. We then impose cross-menu consistency to recover a single influence parameter and a Luce utility representation. 

Fix a feasible set $S$ and interpret $q_S$ as the influencer's total attention to feasible alternatives. Our characterization proceeds in two steps. First, we hold $q_S$ fixed and ask how feasible choice responds to changes in the influencer's feasible behavior, controlling for the amount of unaffordable exposure. This delivers a well-defined marginal effect $\lambda(S,q_S)$: within a fixed feasible-share slice, any change in feasible exposure shifts feasible choice in a common direction and at a common rate. Second, we compare these controlled marginal effects across different feasible shares. The key restriction is a disciplined, one-parameter way in which persuasion weakens as the influencer becomes more aspirational (i.e., as $q_S$ falls), which is precisely the content needed to recover a menu-invariant influence parameter.

\subsection{Per-menu characterization}

On each $\mu$-level, the AWLM denominator $(1-\alpha)+\alpha\mu$ is constant, so induced choice should respond linearly to the feasible restriction $q|_S$. The following axiom captures this linearity without presupposing the functional form. 

\begin{axiom}[Controlled Collinearity]\label{ax:A1}
Fix $\mu\in(0,1)$. For any two exposures $q_1,q_2\in\Delta(\mathcal{X})$ with $q_{1,S}=q_{2,S}=\mu$:

\smallskip
\noindent\emph{(i) Intra-aspiration irrelevance.} If $q_1|_S=q_2|_S$, then $p(\cdot\mid S;q_1)=p(\cdot\mid S;q_2)$.

\smallskip
\noindent\emph{(ii) Proportional response.} For all $x,y\in S$,
\begin{equation}\label{eq:A1}
\bigl[p(x \mid S; q_1) - p(x \mid S; q_2)\bigr]\bigl[q_1(y) - q_2(y)\bigr]
=
\bigl[p(y \mid S; q_1) - p(y \mid S; q_2)\bigr]\bigl[q_1(x) - q_2(x)\bigr].
\end{equation}
\end{axiom}

\noindent
Part~(i) restates Lemma~\ref{lem:outside-irrelevance} as a behavioral 
restriction; on a fixed $\mu$-level, the choice rule depends on $q$ only through its feasible restriction $q|_S$. Part~(ii) requires that the choice shift $\Delta p \equiv p(\cdot \mid S; q_1) - p(\cdot \mid S; q_2)$ and the exposure shift $\Delta q|_S \equiv q_1|_S - q_2|_S$ be collinear.

Figure~\ref{fig:controlled-collinearity} illustrates the geometry. Within a fixed $\mu$-level (a horizontal slice of the simplex in the figure), varying the influencer's within-feasible composition $q(\cdot \mid S)$ traces out a path of induced choices. Axiom~\ref{ax:A1}(ii) requires that the difference vector $p_1 - p_2$ be parallel to $q_1|_S - q_2|_S$ for any two points on this path. Equivalently, every $2 \times 2$ submatrix of $[\Delta p \;\; \Delta q|_S]$ has zero determinant.

\begin{figure}[H]
  \centering
  \includegraphics[width=0.8\textwidth]{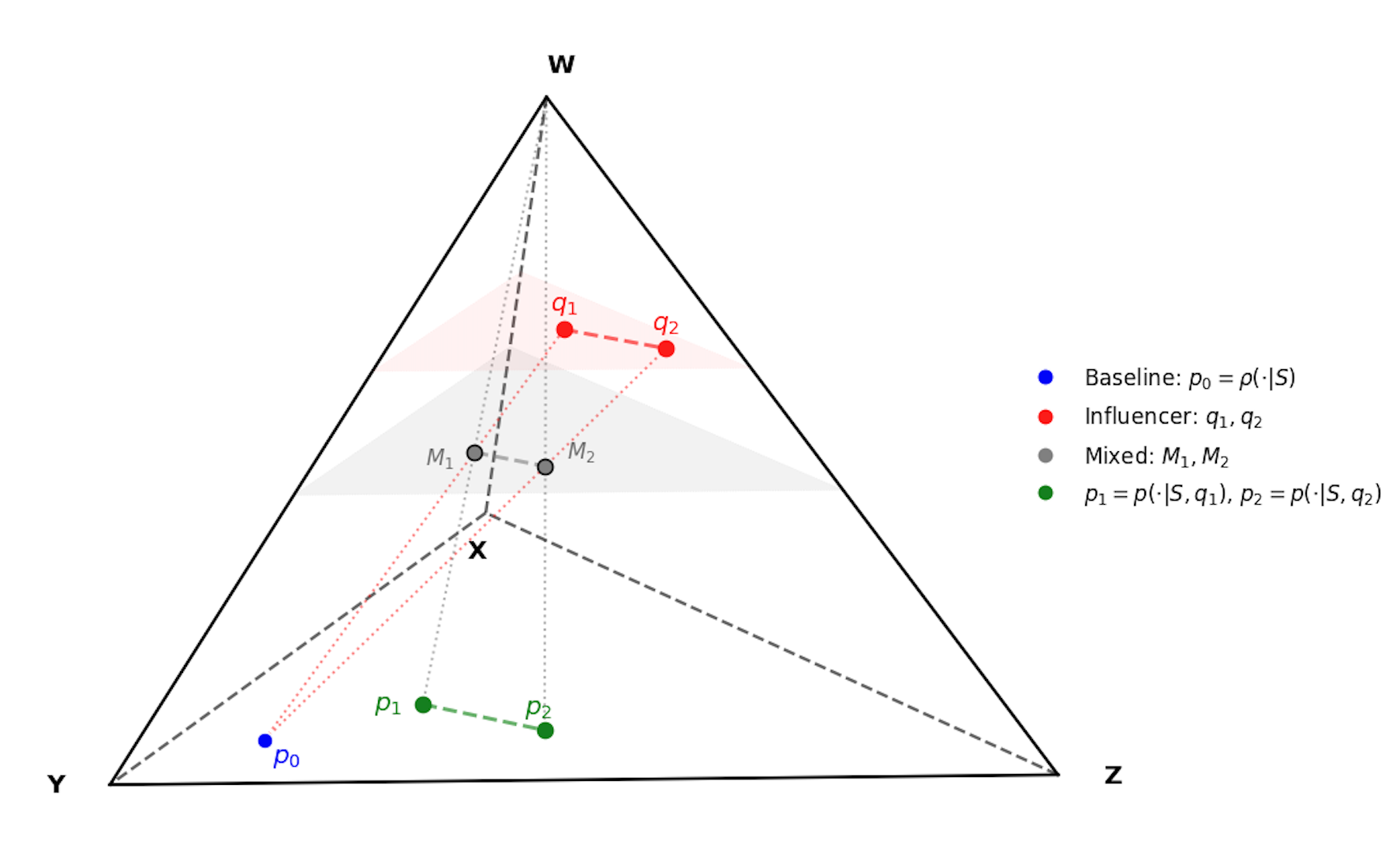}
  \caption{Intuition for Controlled Collinearity (Axiom~\ref{ax:A1}). $S=\{x,y,z\}$; $I=\{x,y,z,w\}$ The influencer's exposures $q_1,q_2$ (red) lie in the interior of $\Delta(I)$. Each exposure mixes with the baseline $p_0$ (blue) to form attempt targets $M_1,M_2$ (gray), which are projected onto $S$ to yield realized choices $p_1,p_2$ (green). Axiom~\ref{ax:A1} requires $p_2-p_1 \parallel (q_2-q_1)|_S$.}
  \label{fig:controlled-collinearity}
\end{figure}

\begin{lemma}[Controlled multiplier and level-affine form]\label{lem:lambda}
Assume $|S|\ge 3$ and Axiom~\ref{ax:A1}. Then for each $\mu\in(0,1)$ there exists a unique scalar
$\lambda(S,\mu)$ such that for all $q_1,q_2\in\Delta(\mathcal{X})$ with $q_{1,S}=q_{2,S}=\mu$,
\begin{equation}\label{eq:lambda-def}
p(\cdot\mid S;q_1)-p(\cdot\mid S;q_2)=\lambda(S,\mu)\bigl(q_1|_S-q_2|_S\bigr).
\end{equation}
Moreover, defining the \emph{level intercept}
\[
a(S,\mu)\equiv   p(\cdot\mid S;q)-\lambda(S,\mu)\,q|_S \qquad \text{for any } q\text{ with }q_S=\mu,
\]
the vector $a(S,\mu)\in\mathbb{R}^S$ is well-defined (independent of the chosen $q$) and
\begin{equation}\label{eq:affine-level}
p(\cdot\mid S;q)=a(S,\mu)+\lambda(S,\mu)\,q|_S \qquad \forall q\text{ with }q_S=\mu.
\end{equation}
\end{lemma}

\begin{proof}
Fix $\mu\in(0,1)$. By Axiom~\ref{ax:A1}(i), within the $\mu$-level the restriction $v\equiv q|_S$ pins down the induced choice.
Write $P(v)$ for the induced choice when $q|_S=v$ and $q_S=\mu$.

Axiom~\ref{ax:A1}(ii) implies that for any $v_1,v_2$ in the $\mu$-level, the vectors $P(v_1)-P(v_2)$ and $v_1-v_2$ are
collinear, so there exists a scalar $\lambda_{12}$ with $P(v_1)-P(v_2)=\lambda_{12}(v_1-v_2)$.

Choose $v_0,v_1,v_2$ in the $\mu$-level such that $(v_1-v_0)$ and $(v_2-v_0)$ are linearly independent (possible since $|S|\ge3$).
Let $P_k\equiv   P(v_k)$. Write
$P_1-P_0=\lambda_{10}(v_1-v_0)$,
$P_2-P_0=\lambda_{20}(v_2-v_0)$, and
$P_1-P_2=\lambda_{12}(v_1-v_2)$.
But also
$P_1-P_2=(P_1-P_0)-(P_2-P_0)=\lambda_{10}(v_1-v_0)-\lambda_{20}(v_2-v_0)$, while
$v_1-v_2=(v_1-v_0)-(v_2-v_0)$.
Linear independence forces $\lambda_{10}=\lambda_{12}=\lambda_{20}$. Denote the common value by $\lambda(S,\mu)$.

The same three-point argument implies that \emph{every} pair $v_1,v_2$ in the level shares this same multiplier, yielding
\eqref{eq:lambda-def}. Finally, define $a(S,\mu)\equiv   P(v)-\lambda(S,\mu)v$ for any $v$ in the level; \eqref{eq:lambda-def}
implies this does not depend on $v$, and \eqref{eq:affine-level} follows.
\end{proof}

\medskip

Geometrically, the multiplier $\lambda(S,\mu)$ governs the strength of the linkage between the influencer's feasible shifts and the DM's reaction. The subsequent axiom restricts how this linkage weakens as the influencer becomes more aspirational.


To state the second part cleanly, define the \emph{normalized level intercept}
\begin{equation}\label{eq:def-pi}
\pi(S,\mu)\equiv   \frac{a(S,\mu)}{\mathbf{1}^\top a(S,\mu)} \in \Delta(S),
\end{equation}
whenever $\mathbf{1}^\top a(S,\mu)>0$.

\begin{axiom}[Leverage Line and Radial Consistency]\label{ax:A2}
Let $\lambda(S, \mu)$ be defined as in Lemma~\ref{lem:lambda}.
\begin{enumerate}[label=\textup{(\roman*)}, leftmargin=2em]
\item \textbf{Leverage line} For all $\mu_1,\mu_2\in(0,1)$,
\begin{equation}\label{eq:A2i}
\frac{1}{\lambda(S,\mu_1)}-\frac{1}{\lambda(S,\mu_2)}=\mu_1-\mu_2,
\qquad\text{and}\qquad
\frac{1}{\lambda(S,\mu)}>\mu\ \ \forall \mu\in(0,1).
\end{equation}

\item \textbf{Radial consistency} For any two exposures $q_1,q_2\in\Delta(\mathcal X)$ with
$q_{1,S}>0$ and $q_{2,S}>0$, if they have the same within-feasible composition,
$q_1(\cdot\mid S)=q_2(\cdot\mid S)$, let $v\equiv q_1(\cdot\mid S)\in\Delta(S)$ denote this common composition.
Then $p(\cdot\mid S;q_1)$, $p(\cdot\mid S;q_2)$, and $v$ are collinear. Equivalently, for all $x,y\in S$,
\begin{equation}\label{eq:A2ii_radial}
\bigl[p(x\mid S;q_1)-v(x)\bigr]\bigl[p(y\mid S;q_2)-v(y)\bigr]
=
\bigl[p(y\mid S;q_1)-v(y)\bigr]\bigl[p(x\mid S;q_2)-v(x)\bigr].
\end{equation}
\end{enumerate}
\end{axiom}

\noindent
Axiom~\ref{ax:A2}(i) says that $1/\lambda(S,\mu)-\mu$ is constant in $\mu$. Denote this constant by $\kappa(S)>0$, so
\[
\frac{1}{\lambda(S,\mu)}=\mu+\kappa(S),
\qquad\text{equivalently}\qquad
\lambda(S,\mu)=\frac{1}{\mu+\kappa(S)}.
\]
Hence $\lambda(S,\mu)$ is \emph{decreasing} in $\mu$: when the influencer devotes less total mass to the feasible set (smaller $\mu$), the normalization factor is smaller and a given \emph{absolute} shift in the feasible exposure vector $q|_S$ produces a larger shift in
$p(\cdot\mid S;\cdot)$.

For aspirational dampening, however, the economically relevant comparison fixes the
influencer's \emph{within-feasible composition} $q(\cdot\mid S)$ rather than the unnormalized
vector $q|_S$. On the $\mu$-level, $q|_S=\mu\,q(\cdot\mid S)$, so
\[
p(\cdot\mid S;q_1)-p(\cdot\mid S;q_2)
=\underbrace{\mu\,\lambda(S,\mu)}_{=:~\beta(S,\mu)}
\Big(q_1(\cdot\mid S)-q_2(\cdot\mid S)\Big).
\]
The scalar $\beta(S,\mu)=\mu\lambda(S,\mu)=\mu/(\mu+\kappa(S))$ is increasing in $\mu$. Thus, as $\mu$ falls (more aspirational mass), the effect of a given change in $q(\cdot\mid S)$ on feasible choice is dampened. Figure~\ref{fig:leverage} provides more intuition.

\begin{figure}[ht!]
  \centering
  \includegraphics[width=0.85\textwidth]{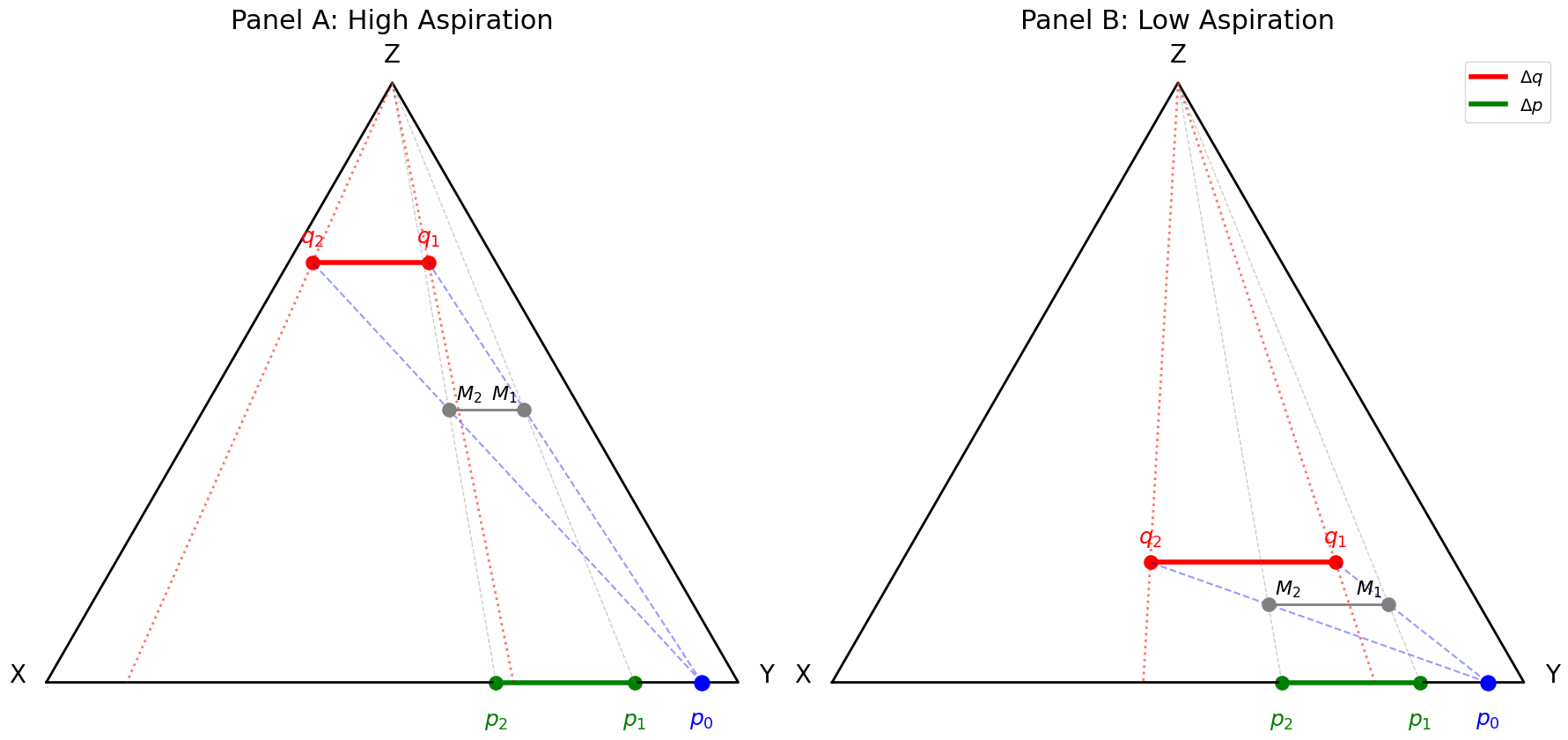}
  \caption{Linking levels via the leverage line (Axiom~\ref{ax:A2}(i)). The panels depict the same within-feasible composition shift evaluated at two aspiration levels $\mu=q_S$. Across $\mu$, the level multiplier $\lambda(S,\mu)$ changes so that $1/\lambda(S,\mu)$ moves one-for-one with $\mu$. Equivalently, the outcome-relevant multiplier for changes in $q(\cdot\mid S)$, namely $\beta(S,\mu)=\mu\lambda(S,\mu)$, is smaller at lower $\mu$ (aspirational dampening).}
  \label{fig:leverage}
\end{figure}
%



\begin{theorem}[Per-menu AWLM Characterization]\label{thm:awlm-permenu}
Fix a nonempty feasible set $S\subseteq\mathcal{X}$ with $|S|\ge3$.
An exposure-dependent choice rule $q\mapsto p(\cdot\mid S;q)$ admits an AWLM representation
\eqref{eq:awlm-permenu} if and only if Axioms~\ref{ax:A1} and~\ref{ax:A2} hold.

Moreover, define $\kappa(S)\equiv   \frac{1}{\lambda(S,\bar\mu)}-\bar\mu$ for any $\bar\mu\in(0,1)$; by Axiom~\ref{ax:A2}(i) this is
well-defined and positive. Then the AWLM parameters are uniquely recovered by
\[
\alpha=\frac{1}{1+\kappa(S)},
\qquad
p_0=\pi(S,\mu)\ \ \text{for any }\mu\in(0,1).
\]
\end{theorem}

\begin{proof}
\textbf{($\Rightarrow$)} Suppose AWLM holds with parameters $(\alpha,p_0)$.
Fix $\mu\in(0,1)$ and take $q_1,q_2$ with $q_{1,S}=q_{2,S}=\mu$. Since the denominator $(1-\alpha)+\alpha\mu$ is common,
\[
p(\cdot\mid S;q_1)-p(\cdot\mid S;q_2)=\frac{\alpha}{(1-\alpha)+\alpha\mu}\,(q_1|_S-q_2|_S),
\]
which implies Axiom~\ref{ax:A1}(i) and~(ii), and identifies
\[
\lambda(S,\mu)=\frac{\alpha}{(1-\alpha)+\alpha\mu}.
\]
Then
\[
\frac{1}{\lambda(S,\mu)}=\frac{(1-\alpha)+\alpha\mu}{\alpha}=\mu+\frac{1-\alpha}{\alpha},
\]
so Axiom~\ref{ax:A2}(i) holds with $\kappa(S)=(1-\alpha)/\alpha>0$.

To verify Axiom~\ref{ax:A2}(ii), take any $q_1,q_2$ with $q_{1,S},q_{2,S}>0$ and
$q_1(\cdot\mid S)=q_2(\cdot\mid S)=:v\in\Delta(S)$. Using the mixture representation
$p(\cdot\mid S;q,\alpha)=(1-\beta(q_S))p_0+\beta(q_S)\,q(\cdot\mid S)$, we have
\[
p(\cdot\mid S;q_i,\alpha)=(1-\beta(q_{i,S}))p_0+\beta(q_{i,S})v,\qquad i=1,2,
\]
hence
\[
p(\cdot\mid S;q_i,\alpha)-v=(1-\beta(q_{i,S}))(p_0-v),
\]
so $p(\cdot\mid S;q_1,\alpha)$, $p(\cdot\mid S;q_2,\alpha)$, and $v$ are collinear, establishing
Axiom~\ref{ax:A2}(ii).

\medskip
\textbf{($\Leftarrow$)} Assume Axioms~\ref{ax:A1} and~\ref{ax:A2}.
By Lemma~\ref{lem:lambda}, for each $\mu\in(0,1)$ we have the level-affine form:
\[
p(\cdot\mid S;q)=a(S,\mu)+\lambda(S,\mu)\,q|_S \qquad \forall q\text{ with }q_S=\mu.
\]
Summing over $S$ gives, for any $q$ with $q_S=\mu$,
\[
1=\mathbf{1}^\top a(S,\mu)+\lambda(S,\mu)\mu
\quad\Longrightarrow\quad
\mathbf{1}^\top a(S,\mu)=1-\mu\lambda(S,\mu).
\]
By Axiom~\ref{ax:A2}(i), $1/\lambda(S,\mu)=\mu+\kappa(S)$ for some $\kappa(S)>0$, hence
\[
\mathbf{1}^\top a(S,\mu)=1-\frac{\mu}{\mu+\kappa(S)}=\frac{\kappa(S)}{\mu+\kappa(S)}>0.
\]

\medskip
\noindent\emph{Step 1: $a(S,\mu)\in\mathbb R^S_+$ and $\pi(S,\mu)$ is well-defined.}
Fix $\mu\in(0,1)$ and any $y\in S$. Choose an $x\in S$ with $x\neq y$ and consider an exposure $q$ with
$q_S=\mu$ and $q(\cdot\mid S)=\delta_x$ (equivalently $q|_S=\mu\,\delta_x$), where $\delta_x\in\Delta(S)$ is the degenerate
distribution on $x$. Then
\[
p(\cdot\mid S;q)=a(S,\mu)+\mu\lambda(S,\mu)\,\delta_x.
\]
In particular, for $y\neq x$, we have $p(y\mid S;q)=a(S,\mu)(y)\ge 0$. Since $y$ was arbitrary, this shows
$a(S,\mu)\ge 0$ componentwise. Together with $\mathbf 1^\top a(S,\mu)>0$, the normalized intercept
$\pi(S,\mu)\equiv a(S,\mu)/(\mathbf 1^\top a(S,\mu))$ is well-defined and lies in $\Delta(S)$.

\medskip
\noindent\emph{Step 2: Radial consistency implies intercept-direction stability.}
Fix $\mu_1,\mu_2\in(0,1)$ and write $a_i\equiv a(S,\mu_i)$ and $\lambda_i\equiv\lambda(S,\mu_i)$.
Pick $y_0\in S$ with $a_1(y_0)>0$ (exists because $a_1\ge 0$ and $\mathbf 1^\top a_1>0$).
Also pick $z\in S$ with $a_2(z)>0$.

Choose $x\in S\setminus\{y_0,z\}$ (possible since $|S|\ge 3$). For each $i=1,2$, take an exposure $q_i$ with
$q_{i,S}=\mu_i$ and $q_i(\cdot\mid S)=\delta_x$. Then
\[
p(\cdot\mid S;q_i)=a_i+\mu_i\lambda_i\,\delta_x.
\]
Since $y_0,z\neq x$, we have $p(y_0\mid S;q_i)=a_i(y_0)$ and $p(z\mid S;q_i)=a_i(z)$.
Applying Axiom~\ref{ax:A2}(ii) (radial consistency) to the common composition $v=\delta_x$ and the pair $(q_1,q_2)$ yields
\[
a_1(y_0)a_2(z)=a_1(z)a_2(y_0).
\]
The left-hand side is strictly positive, so $a_2(y_0)>0$. Define $t\equiv a_2(y_0)/a_1(y_0)>0$.

Now fix any $y\in S$. Choose $x\in S\setminus\{y,y_0\}$ (again possible since $|S|\ge 3$) and repeat the argument with
$v=\delta_x$ and coordinates $(y,y_0)$ to obtain
\[
a_1(y)a_2(y_0)=a_1(y_0)a_2(y),
\]
hence $a_2(y)=t\,a_1(y)$. Since $y$ was arbitrary, $a_2=t\,a_1$, which implies
\[
\pi(S,\mu_1)=\pi(S,\mu_2).
\]
Because $\mu_1,\mu_2$ were arbitrary, $\pi(S,\mu)$ is constant in $\mu$; denote the common element by $p_0\in\Delta(S)$.

\medskip
\noindent\emph{Step 3: Recover the AWLM form.}
By Step 2 and $\mathbf 1^\top a(S,\mu)=\kappa(S)/(\mu+\kappa(S))$, we have
\[
a(S,\mu)=(\mathbf 1^\top a(S,\mu))\,p_0=\frac{\kappa(S)}{\mu+\kappa(S)}\,p_0.
\]
Plugging into \eqref{eq:affine-level} with $\mu=q_S$ yields
\[
p(\cdot\mid S;q)
=\frac{\kappa(S)}{\kappa(S)+q_S}\,p_0+\frac{1}{\kappa(S)+q_S}\,q|_S
=\frac{\kappa(S)p_0+q|_S}{\kappa(S)+q_S}.
\]
Let $\alpha\equiv \frac{1}{1+\kappa(S)}\in(0,1)$ so that $\kappa(S)=(1-\alpha)/\alpha$, and multiply numerator and denominator by $\alpha$ to obtain
\[
p(\cdot\mid S;q)=\frac{(1-\alpha)p_0+\alpha q|_S}{(1-\alpha)+\alpha q_S},
\]
which is exactly the AWLM representation \eqref{eq:awlm-permenu}.
\end{proof}

\begin{remark}[Linking back to Figures~\ref{fig:controlled-collinearity}--\ref{fig:leverage}.]
Axiom~\ref{ax:A2}(i) links levels by imposing $1/\lambda(S,\mu)=\mu+\kappa(S)$. While this implies $\lambda(S,\mu)$ rises as $\mu$ falls (a normalization effect), the outcome-relevant multiplier for within-feasible composition is $\beta(S,\mu)=\mu\lambda(S,\mu)$, which falls as $\mu$ falls (aspirational dampening).
\end{remark}

\subsection{Cross-Menu Characterization}
\label{subsec:global-char}

Theorem~\ref{thm:awlm-permenu} characterizes AWLM behavior on a fixed feasible set $S$
with $|S|\ge 3$. It recovers (i) an \emph{intrinsic} choice distribution $p_0^S\in\Delta(S)$
and (ii) a leverage-line intercept $\kappa(S)>0$ from Axiom~\ref{ax:A2}(i).
It is convenient to reparameterize this intercept by the menu-level influence index
\[
\alpha_S \equiv    \frac{1}{1+\kappa(S)}\in(0,1)
\qquad\Longleftrightarrow\qquad
\kappa(S)=\frac{1-\alpha_S}{\alpha_S}.
\]
To obtain a global AWLM with a \emph{single} influence parameter and a \emph{single} Luce utility $u$
across menus, we impose two consistency requirements:
(i)~the recovered $\alpha_S$ is common across feasible sets, and
(ii)~the recovered intrinsic distributions are Luce-consistent across overlapping menus.

For each nonempty $S\subseteq\mathcal{X}$ with $|S|\ge 3$, suppose the per-menu characterization holds and yields some $(p_0^S,\alpha_S)$ such that
\begin{equation}
\label{eq:AWLM-intrinsic-family}
p(x \mid S; q) = \frac{(1-\alpha_S)p_0^S(x) + \alpha_S q(x)}{(1-\alpha_S) + \alpha_S q_S},
\qquad x \in S.
\end{equation}

\begin{definition}[Luce-consistent intrinsic family]
\label{def:Luce-intrinsic}
A family $\{p_0^S\}_{|S|\ge 3}$ is \emph{Luce-consistent} if for all nonempty $S,T\subseteq\mathcal{X}$
(with $|S|,|T|\ge 3$) and all $x,y\in S\cap T$,
\[
\frac{p_0^S(x)}{p_0^S(y)}=\frac{p_0^T(x)}{p_0^T(y)}.
\]
\end{definition}

By \citet{Luce_1959}'s representation theorem, Luce-consistency is equivalent to the existence of
$u:\mathcal{X}\to\mathbb{R}_{++}$ such that for every $S$ with $|S|\ge 3$ and every $x\in S$,
\[
p_0^S(x)=\frac{u(x)}{\sum_{z\in S}u(z)}.
\]

\begin{theorem}[Global AWLM Characterization (menus with $|S|\ge 3$)]
\label{thm:awlm-global}
A family $\{p(\cdot \mid S; q)\}_{S\subseteq\mathcal X,\ |S|\ge 3,\ q\in\Delta(\mathcal X)}$
is generated by an Aspiration-Weighted Luce Model with some $\alpha\in(0,1)$ and
$u:\mathcal X\to\mathbb R_{++}$ if and only if:
\begin{enumerate}[label=(\roman*)]
\item For each nonempty $S\subseteq\mathcal X$ with $|S|\ge 3$, the per-menu rule
$q\mapsto p(\cdot\mid S;q)$ satisfies Axioms~\ref{ax:A1} and~\ref{ax:A2}, and the recovered
menu-level influence indices $\alpha_S$ are constant across such $S$ (equivalently, $\kappa(S)$ is constant).
Denote the common value by $\alpha$; and
\item The induced intrinsic family $\{p_0^S\}_{|S|\ge 3}$ recovered from Theorem~\ref{thm:awlm-permenu}
is Luce-consistent on overlaps.
\end{enumerate}
In that case, for every $S$ with $|S|\ge 3$, every $q$, and every $x\in S$,
\[
p(x \mid S; q)=\frac{(1-\alpha)\frac{u(x)}{\sum_{z \in S} u(z)}+\alpha q(x)}{(1-\alpha)+\alpha q_S}.
\]
Moreover, once $(\alpha,u)$ are recovered from menus with $|S|\ge 3$, the same formula
uniquely extends the model to all nonempty menus.
\end{theorem}

\begin{remark}[Small menus]
\label{rem:small-menus}
For $|S|=1$, the formula gives $p(x|\{x\};q)=1$ for any $q$, which is trivially satisfied.
For $|S|=2$, the formula yields non-trivial predictions, but the per-menu axioms used in
Theorem~\ref{thm:awlm-permenu} (notably the three-point argument behind Lemma~\ref{lem:lambda})
do not apply. The global characterization therefore relies on menus with $|S|\ge 3$ to pin down
$(\alpha,u)$; binary menus provide consistency checks rather than identifying restrictions.
In contrast, the identification results in Section~\ref{sec:identification} apply whenever $|S|\ge 2$.
\end{remark}

\begin{proof}
\textbf{($\Rightarrow$)}
If the AWLM holds with parameters $(\alpha,u)$, then for every menu $S, |S|\geq 3$ the induced intrinsic distribution is Luce\footnote{And hence for all smaller menus by extension}:
\[
p_0^S(x)=\frac{u(x)}{\sum_{z\in S}u(z)}.
\]
Substituting this into \eqref{eq:AWLM-intrinsic-family} yields the AWLM formula on each $S$.
In particular, for every $S$ with $|S|\ge 3$ the per-menu rule satisfies Axioms~\ref{ax:A1}--\ref{ax:A2}
and the recovered influence parameter equals the common $\alpha$. Luce-consistency on overlaps follows immediately because odds ratios equal $u(x)/u(y)$.

\medskip
\textbf{($\Leftarrow$)}
Assume (i) and (ii). By (i), on each $S$ with $|S|\ge 3$ the per-menu characterization yields a
representation of the form \eqref{eq:AWLM-intrinsic-family} with a common $\alpha$ and some
intrinsic distribution $p_0^S\in\Delta(S)$.
By (ii) and Definition~\ref{def:Luce-intrinsic}, the family $\{p_0^S\}_{|S|\ge 3}$ is Luce-consistent on overlaps,
so there exists $u:\mathcal X\to\mathbb R_{++}$ such that $p_0^S(x)=u(x)/\sum_{z\in S}u(z)$ for all $|S|\ge 3$.
Substituting this into \eqref{eq:AWLM-intrinsic-family} gives the stated AWLM formula on all menus with $|S|\ge 3$.
\end{proof}

Per-menu, AWLM is pinned down by a fixed-$\mu$ restriction implying level-affine responses with a common multiplier $\lambda(S,\mu)$, and a cross-$\mu$ restriction that forces the level multipliers to satisfy the slope-one relation $1/\lambda(S,\mu)=\mu+\kappa(S)$ while the normalized intercept direction is invariant in $\mu$.

Globally, Luce-consistency ties the recovered intrinsic distributions across menus into a single utility $u$. (Once $(\alpha,u)$ are recovered from menus with $|S|\ge 3$, the same formula provides a unique extension to smaller menus.)

\section{Identification}
\label{sec:identification}
We have presented the AWLM as a mapping from primitives to observable choice probabilities. This section asks when and how these structural parameters can be recovered from observed stochastic choice data. Recovering $\alpha$ allows us to quantify influence strength across contexts, while recovering idiosyncratic preferences $p_0(\cdot\mid S)$ enables counterfactual analysis and policy evaluation.

This section establishes that the AWLM parameters $(\alpha, p_0)$ are point-identified from observed stochastic choice under minimal exposure variation. The key result is a simple linear identity: when the influencer's distribution shifts from $q_1$ to $q_2$, the induced change in feasible choices satisfies $\Delta = \alpha A$, where $\Delta$ and $A$ are observable. This provides constructive identification from as few as two exposure conditions and generates testable restrictions when additional exposures are available.

\begin{figure}[ht!]
  \centering
  \includegraphics[width=0.65\textwidth]{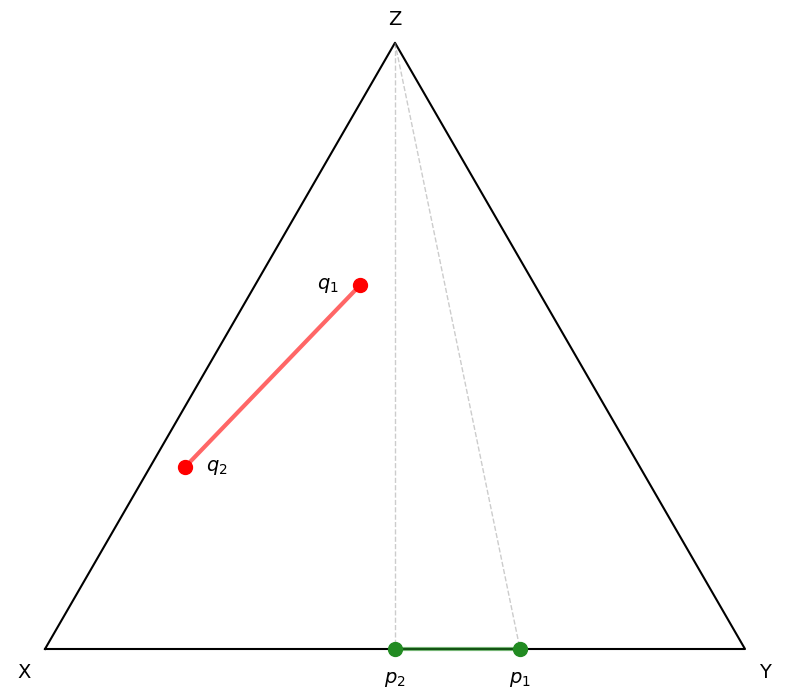}
  \caption{Identification Problem: The influencer shifts her exposure distribution (red segment), inducing a corresponding shift in the DM's feasible choices (green segment). The identification goal is to decode the unobservable influence strength from the geometric relationship between these observations.}
  \label{fig:id-geometry}
\end{figure}


\subsection{Two exposures on a fixed feasible set}
\label{subsec:two-exposures}
Fix a feasible set $S\subseteq \mathcal{X}$ with $|S| \geq 2$. For exposures $q_1, q_2 \in \Delta(I)$ with $I \supseteq S$, let $P_i \equiv p(\cdot \mid S; q_i, \alpha) \in \Delta(S)$ denote the induced choice distributions. Define the normalization factor $D_i(\alpha) \equiv (1-\alpha) + \alpha q_{i,S}$.

The AWLM representation \eqref{eq:awlm-closed} can be rearranged to yield the \emph{pre-normalization identity}:
\begin{equation}
D_i(\alpha) P_i(x) = (1-\alpha) p_0(x \mid S) + \alpha q_i(x), \qquad x \in S.
\label{eq:prenorm-main}
\end{equation}
Subtracting this identity across $i = 1, 2$ eliminates the idiosyncratic term:
\[
D_2(\alpha) P_2(x) - D_1(\alpha) P_1(x) = \alpha \bigl( q_2(x) - q_1(x) \bigr).
\]
To express this in terms of observables, define
\begin{equation}
\Delta \equiv P_2 - P_1, \qquad
A \equiv (q_2 - q_1)|_S + (1 - q_{2,S}) P_2 - (1 - q_{1,S}) P_1.
\label{eq:A-def}
\end{equation}

\begin{proposition}[Two-exposure identity]
\label{prop:two-exposure-identity}
Under AWLM, for any two exposures $q_1,q_2$ on $I\supseteq S$ we have
\begin{equation}
\Delta = \alpha A.
\label{eq:Delta-equals-alphaA}
\end{equation}
Conversely, suppose \eqref{eq:Delta-equals-alphaA} holds for some $\alpha\in(0,1)$. Define $p_0\in\mathbb{R}^S$ by
\[
(1-\alpha)p_0(x) \equiv    D_1(\alpha)P_1(x) - \alpha q_1(x),\qquad x\in S.
\]
Then \eqref{eq:prenorm-main} holds for $i=1$ by construction, and substituting $P_2=P_1+\Delta$ with $\Delta=\alpha A$ shows it also holds for $i=2$ with the same $p_0$. Moreover, if $p_0\in\Delta(S)$, it can be interpreted as the idiosyncratic distribution $p_0(\cdot\mid S)$.
\end{proposition}

\begin{proof}
See Appendix~\ref{app:proofs-identification}.
\end{proof}

The identity $\Delta = \alpha A$ states that the observable choice shift is proportional to an observable function of the exposure shift, with $\alpha$ as the proportionality constant. Proposition~\ref{prop:two-exposure-identity} immediately yields point identification of $\alpha$ whenever there is genuine variation in exposure.

\begin{corollary}[Identification of $\alpha$ from two exposures]
\label{cor:alpha-two}
Suppose AWLM holds on $S$ and consider two exposures $q_1,q_2$ on $I\supseteq S$ with corresponding choices $P_1,P_2\in\Delta(S)$. If $A(x)\neq 0$ for some $x\in S$, then
\begin{equation}
\alpha = \frac{\Delta(x)}{A(x)},
\label{eq:alpha-ratio}
\end{equation}
and this ratio is the same for all $x\in S$ with $A(x)\neq 0$.\footnote{Figure~\ref{fig:degenerate} in Appendix~\ref{app:figures} illustrates the identification geometry: Panel~A shows a design with $A\neq 0$ that identifies $(\alpha,p_0)$; Panel~B shows the knife-edge case where $A=0$ and $\alpha$ is unidentified; Panel~C depicts a near-degenerate design leading to weak identification.}
\end{corollary}

\begin{remark}[Special case: equal total mass on $S$]
\label{rem:qS-equal}
If $q_{1,S}=q_{2,S}=:q_S$, the normalization factors coincide and the identity simplifies. In this case, $A$ is collinear with $(q_2-q_1)|_S$, and substituting $\Delta=\alpha A$ yields
\[
\frac{P_2(x)-P_1(x)}{q_2(x)-q_1(x)} = \frac{\alpha}{(1-\alpha)+\alpha q_S} =:\lambda(S,q_S),
\]
independent of $x$. This gives the explicit formula $\alpha = \frac{\lambda(S,q_S)}{1+\lambda(S,q_S)(1-q_S)}$, recovering the proportional-difference condition from Corollary~\ref{cor:alpha-two} as a special case.
\end{remark}

Once $\alpha$ has been identified, the idiosyncratic Luce distribution $p_0(\cdot\mid S)$ is recovered by inverting the pre-normalization identity \eqref{eq:prenorm-main}.

\begin{proposition}[Identification of the idiosyncratic distribution on $S$]
\label{prop:idiosyncratic-id}
Suppose AWLM holds on $S$ and $\alpha\in(0,1)$ has been identified from two exposures $q_1,q_2$. For any exposure $q\in\Delta(\mathcal{X})$ define
\begin{equation}
p_0^S(x)
\equiv   
\frac{D(S,q;\alpha)p(x;S,q,\alpha) - \alpha q(x)}{1-\alpha},
\qquad x\in S.
\label{eq:P0-def}
\end{equation}
Then $p_0^S\in\Delta(S)$ is independent of the chosen exposure $q$ and coincides with the idiosyncratic Luce distribution:
\[
p_0^S(x) = p_0(x\mid S) = \frac{u(x)}{\sum_{y\in S}u(y)},\qquad x\in S.
\]
Moreover, the Luce weights $(u(x)){x \in S}$ are identified up to a common positive multiplicative factor. Recovering a single global utility $u: \mathcal{X} \to \mathbb{R}_{++}$ requires the cross-menu structure developed in Section~\ref{subsec:global-char}.
\end{proposition}

Figure~\ref{fig:identification-2d} illustrates the geometric intuition: the true $p_0$ is uniquely pinned down by requiring that the attempt targets $M_1,M_2$ satisfy the parallelism $M_2-M_1 \parallel q_2-q_1$.

\paragraph{Identification procedure.}
Given two exposure--choice pairs $(q_1,P_1)$ and $(q_2,P_2)$ on $S$: 
compute $\Delta=P_2-P_1$ and $A$ as in \eqref{eq:A-def}; if $A(x)\neq 0$ for some $x$, then $\alpha=\Delta(x)/A(x)$ (the ratio is constant across all such $x$); finally, recover $p_0(\cdot\mid S)$ from \eqref{eq:P0-def} using either exposure.

\begin{figure}[H]
  \centering
  \includegraphics[width=0.9\textwidth]{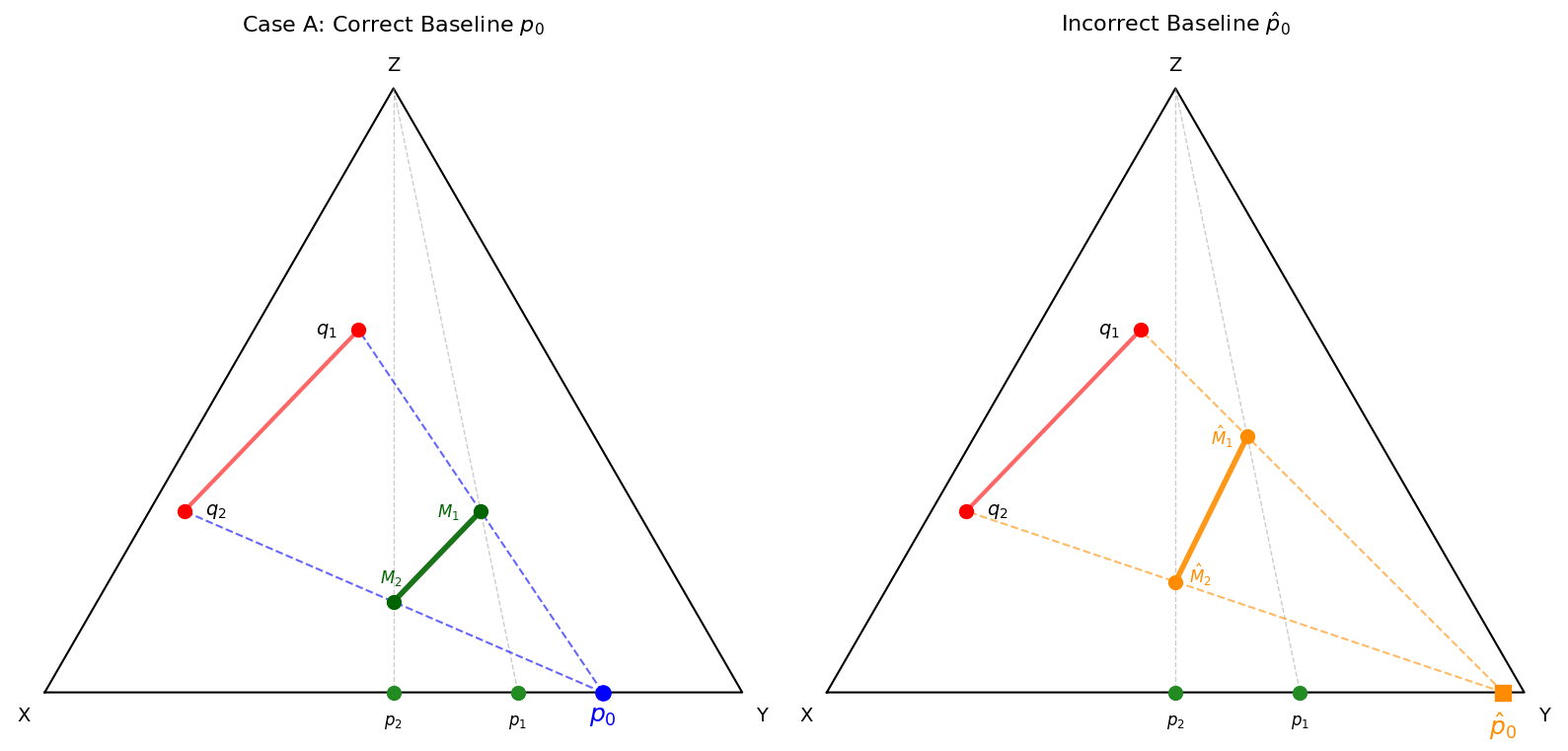}
  \caption{Geometric intuition for identification of $p_0(\cdot\mid S)$: Each attempt target $M_i$ lies at the intersection of two geometric objects: the chord from $p_0$ to $q_i$ (mixing) and the ray from $Z$ through $p_i$ (conditioning). The AWLM structure requires $M_2 - M_1 = \alpha(q_2-q_1)$, so the segment $M_1 M_2$ must be \textbf{parallel} to $q_1 q_2$.  Left panel: the true $p_0$ produces targets satisfying this parallelism. Right panel: an incorrect baseline $\hat{p}_0$ yields targets $\hat{M}_1, \hat{M}_2$ (orange) whose connecting segment is \emph{not} parallel to the exposure shift, violating the model's structure.}
  \label{fig:identification-2d}
\end{figure}

\subsection{Noisy shares: minimum distance and GMM}
\label{subsec:gmm-main}

The identity $\Delta=\alpha A$ is a population restriction. In samples, both $\widehat\Delta$ and $\widehat A$ are constructed from the same empirical shares, inducing mechanical correlation between the regressor and the regression error in a naive least-squares implementation. We therefore estimate $(\alpha,p_0)$ from the primitive pre-normalization moments \eqref{eq:prenorm-main}, which yields a standard minimum-distance/GMM procedure and delivers overidentification tests when multiple exposure regimes are observed. 

Fix a feasible set $S$ with $|S|=m\ge 2$ and exposures $\{q_i\}_{i=1}^K$ on $I\supseteq S$. For $\alpha\in(0,1)$ define $D_i(\alpha)\equiv   (1-\alpha)+\alpha q_{i,S}$, where $q_{i,S}\equiv   \sum_{x\in S}q_i(x)$. Motivated by \eqref{eq:prenorm-main}, define the $m$-dimensional moment
\[
m_i(\alpha,p_0)\equiv   D_i(\alpha)\hat P_i-(1-\alpha)p_0-\alpha\,q_i|_S\in\mathbb R^S.
\]
Under AWLM, $\mathbb E[m_i(\alpha_0,p_0)]=0$ for all $i$.

To avoid the adding-up singularity inherent in multinomial data, we drop one alternative and work in $\mathbb R^{m-1}$. Let $H$ be the selector matrix that keeps the first $m-1$ coordinates, and write $\tilde m_i(\alpha,\tilde p_0)\equiv   H m_i(\alpha,p_0)$. Stack $\tilde m(\theta)\equiv   (\tilde m_1^\top,\ldots,\tilde m_K^\top)^\top$ with $\theta\equiv   (\alpha,\tilde p_0)$, and define the GMM criterion
\[
Q_N(\theta;W)\equiv   N\,\tilde m(\theta)^\top W\,\tilde m(\theta),
\]
for a positive definite weight matrix $W$ (block-diagonal in $i$). The estimator is $\hat\theta\in\arg\min_{\theta}Q_N(\theta;W)$.

For fixed $\alpha$, the criterion is quadratic in $\tilde p_0$ and can be concentrated out. When $W=\mathrm{blkdiag}(W_1,\ldots,W_K)$,
\[
\hat{\tilde p}_0(\alpha)
=
\frac{1}{1-\alpha}\Big(\sum_{i=1}^K W_i\Big)^{-1}\sum_{i=1}^K W_i\,H\!\left(D_i(\alpha)\hat P_i-\alpha q_i|_S\right).
\]
We then minimize the one-dimensional concentrated objective $Q_N^c(\alpha;W)\equiv   Q_N((\alpha,\hat{\tilde p}_0(\alpha));W)$ over $\alpha\in(0,1)$.

A convenient one-step choice is $W=I$, which yields a simple minimum-distance estimator. For efficient estimation and inference, we use the optimal weight $W=\hat\Omega^{-1}$ based on the multinomial variance of $\hat P_i$ (see Appendix~\ref{app:gmm} for details) and report the standard $J$-test of overidentifying restrictions when $K(m-1)>m$.

\begin{remark}[Pairwise least squares as a diagnostic]
\label{rem:ls-diagnostic}
For a single pair of exposures, the identity $\Delta=\alpha A$ suggests a closed-form estimator $\hat\alpha=\langle\Delta,A\rangle/\langle A,A\rangle$. While this can be useful as a quick diagnostic check, it inherits the generated-regressor problem noted above and can be unstable when $A\approx 0$. The GMM approach avoids these issues by working directly with the pre-normalization moments.
\end{remark}

\begin{example}[Minimum distance with three exposures]
\label{ex:gmm-noisy}
Fix $S=\{x,y,z\}$ and suppose we observe three exposure regimes with feasible restrictions
\[
q_1|_S=(0.4,0.2,0.1),\quad q_2|_S=(0.3,0.1,0.2),\quad q_3|_S=(0.25,0.25,0.25),
\]
and feasible shares $(q_{1,S},q_{2,S},q_{3,S})=(0.7,0.6,0.75)$.
From $N_i=60$ independent observations in each regime we obtain empirical shares
\[
\hat P_1=\Big(\tfrac{19}{60},\tfrac{18}{60},\tfrac{23}{60}\Big),\quad
\hat P_2=\Big(\tfrac{18}{60},\tfrac{15}{60},\tfrac{27}{60}\Big),\quad
\hat P_3=\Big(\tfrac{15}{60},\tfrac{19}{60},\tfrac{26}{60}\Big).
\]

We estimate $(\alpha,p_0)$ by the one-step minimum-distance criterion with moments $m_i(\alpha,p_0)=D_i(\alpha)\hat P_i-(1-\alpha)p_0-\alpha\,q_i|_S$ and weight $W=I$. The minimizer is
\[
\hat\alpha\approx 0.41,\qquad \hat p_0\approx(0.206,\,0.300,\,0.494).
\]
With $K=3$ and $|S|=3$, the number of moment conditions is $K(|S|-1)=6$ while the parameter dimension is $|S|=3$, so the model is overidentified with $\mathrm{df}=3$. The same moments support a standard $J$-test once $W$ is replaced by the optimal multinomial weight.
\end{example}


\subsection{Generic identification \`a la McManus}
\label{subsec:mcmanus}

Propositions~\ref{prop:two-exposure-identity}--\ref{prop:idiosyncratic-id} provide
constructive, finite-sample identification on a fixed feasible set $S$:
a single nondegenerate pair of exposures already point-identifies the influence strength and the idiosyncratic distribution. It is nevertheless useful to ask whether identification
is robust from a design perspective, in the topological sense of \citet{mcmanus1992common}.
Fix a feasible set $S$ with $|S|=m\ge 2$ and an influencer menu $I\supsetneq S$.
View an exposure profile $q=(q_1,\ldots,q_K)\in\Delta(I)^K$ as the analyst's design choice.
For each design $q$, the AWLM implies a smooth map
\[
\Phi_q:(0,1)\times\Delta^{\circ}(S)\to(\Delta(S))^K
\quad\text{given by}\quad
\Phi_q(\alpha,p_0)\equiv   \big(P_1,\ldots,P_K\big),
\]
where $P_i(\cdot)=p(\cdot\mid S;q_i,\alpha)$ is the induced choice distribution under exposure $q_i$.

When $K\ge 3$, globally- and locally-identifying designs are generic:
there is an open dense set of exposure profiles for which $\Phi_q$ is injective
and has full column-rank Jacobian everywhere on $(0,1)\times\Delta^{\circ}(S)$.
Intuitively, failure of identification requires a knife-edge exposure design:
the within-$S$ restriction vectors $\{q_i|_S\}_{i=1}^K$ must move along a single affine line,
so that induced choices can collapse to a common point for some parameter values.
Such collinearity is destroyed by arbitrarily small perturbations of the exposure profile,
making identification robust in the generic sense.
A formal proof is provided in Appendix~\ref{app:proof-generic-id}.

\begin{proposition}[Generic identification with $K\ge 3$]
\label{prop:generic-McManus}
Fix a feasible set $S$ with $|S|=m\ge 2$, an influencer menu $I\supsetneq S$, and $K\ge 3$.
There exists an open dense set $\mathcal Q^\ast\subset\Delta(I)^K$ such that,
for every exposure profile $q=(q_1,\ldots,q_K)\in\mathcal Q^\ast$,
the map $\Phi_q$ is injective and its Jacobian with respect to $(\alpha,p_0)$
has full column rank $m$ at every $(\alpha,p_0)\in(0,1)\times\Delta^{\circ}(S)$.
In particular, for such generic designs, the AWLM parameters $(\alpha,p_0)$
are globally and locally identified from $(P_1,\ldots,P_K)$.
\end{proposition}

\subsection{Multiple exposures and overidentification}
\label{subsec:multi-exposure}

Now suppose we observe $K\ge 2$ distinct exposures $\{q_i\}_{i=1}^K$ on $I\supseteq S$ with associated choice probabilities $P_i(\cdot)\in\Delta(S)$.
For each $i$, let $q_{i,S}\equiv   \sum_{x\in S}q_i(x)$ and $D_i(\alpha)=(1-\alpha)+\alpha q_{i,S}$, and for each pair $(i,j)$ define
\[
\Delta_{ij} \equiv    P_j - P_i,\qquad
A_{ij} \equiv   (q_j-q_i)\big|_S + (1-q_{j,S})P_j - (1-q_{i,S})P_i .
\]

\begin{proposition}[Finite-sample AWLM rationalizability]
\label{prop:finite-rationalizability}
A finite dataset $\{(q_i,P_i)\}_{i=1}^K$ on $S$ is generated by a common AWLM with some
$(\alpha,p_0(\cdot\mid S))\in(0,1)\times\Delta(S)$ if and only if there exists $\alpha\in(0,1)$ such that
\begin{equation}
\label{eq:pairwise-all}
\Delta_{ij}=P_j-P_i=\alpha A_{ij}\qquad \forall\,1\le i<j\le K,
\end{equation}
and, for some (equivalently every) $i$,
\[
b_i(\alpha):=D_i(\alpha)P_i-\alpha\,q_i|_S \in \mathbb R^S_+,
\qquad D_i(\alpha):=(1-\alpha)+\alpha q_{i,S}.
\]
In that case $b_1(\alpha)=\cdots=b_K(\alpha)=(1-\alpha)p_0(\cdot\mid S)$ and
$p_0(\cdot\mid S)$ is unique and equals $b_i(\alpha)/(1-\alpha)$.
\end{proposition}

Proposition~\ref{prop:finite-rationalizability} shows that additional exposures beyond $K=2$ generate \emph{overidentifying restrictions}: each pair $(i,j)$ must yield the same $\alpha$ and the same idiosyncratic distribution vector $p_0^S$. This provides a natural falsification test for AWLM on a given menu.

\section{Conclusion}
\label{sec:conclusion}
The Aspiration-Weighted Luce Model addresses the question of how social influence operates when the influencer consumes from a richer menu than the follower. Our approach is predicated on a specific order of operations where influence is formed over the full menu and feasibility is imposed afterward through normalization. This mechanism generates a distinctive prediction of \emph{aspirational dampening}, where shifting the influencer's attention toward inaccessible alternatives weakens her effective impact on feasible consumption, even if her within-feasible behavior remains unchanged.

We provided a behavioral characterization that disentangles idiosyncratic Luce structure from directed influence, showing that the model's parameters are point-identified from as few as two exposure conditions. The underlying geometry is transparent, requiring that observed choice shifts be parallel to exposure shifts in a manner that uniquely pins down the latent baseline.

Future extensions might consider frameworks to include multiple influencers, endogenous exposure, or network structures. The core insight remains that, when aspirations exceed feasibility, the normalization step systematically reshapes choice in an empirically identifiable way.

\bibliography{Mybib}

\appendix
\section{Notation Glossary}
Throughout the paper, we use the following notation:
\begin{align*}
\X &\text{: a finite universe of alternatives},\\
S\subseteq \X &\text{: the decision maker's feasible set},\\
I\subseteq \X,\ S\subseteq I &\text{: the influencer's menu},\\
u:\X\to\mathbb R_{++} &\text{: Luce weights (idiosyncratic utilities)},\\
p_0(x\mid S)=\frac{u(x)}{\sum_{y\in S}u(y)} &\text{: idiosyncratic Luce choice on }S,\\
\rho^S(x)=p_0(x\mid S)\mathbf{1}\{x\in S\} &\text{: zero-extension of }p_0(\cdot\mid S)\text{ to }\X,\\
q\in\Delta(I) &\text{: exposure distribution (identified with its zero-extension to }\X\text{)},\\
q_S=\sum_{x\in S}q(x) &\text{: feasible share of exposure},\\
\alpha\in(0,1) &\text{: attempt-level influence parameter},\\
p(\cdot\mid S;q,\alpha)\in\Delta(S) &\text{: induced feasible choice under AWLM}.
\end{align*}

\section{Proofs}
\subsection{Proof for Subsection~\ref{subsec:comparative}}
\begin{proof}[Proof of Proposition~\ref{prop:dampening}]
Fix $S$, $p_0(\cdot\mid S)$, and $\alpha\in(0,1)$. Let $q,q'\in\Delta(\mathcal X)$ satisfy
$q(\cdot\mid S)=q'(\cdot\mid S)$ and $q_S,q'_S>0$.

By Equation~\ref{eq:beta-mixture}, for any exposure $\tilde q$ with $\tilde q_S>0$,
\[
p(\cdot\mid S;\tilde q,\alpha)
=\bigl(1-\beta(\tilde q_S)\bigr)p_0(\cdot\mid S)+\beta(\tilde q_S)\,\tilde q(\cdot\mid S).
\]
Apply this to $\tilde q=q$ and $\tilde q=q'$, and use $q(\cdot\mid S)=q'(\cdot\mid S)$:
\begin{align*}
p(\cdot\mid S;q',\alpha)-p(\cdot\mid S;q,\alpha)
&=\bigl(\beta(q'_S)-\beta(q_S)\bigr)\Big(q(\cdot\mid S)-p_0(\cdot\mid S)\Big),
\end{align*}
which is the stated identity.

Because $\beta(\cdot)$ is strictly increasing on $(0,1]$ when $\alpha\in(0,1)$,
$q'_S<q_S$ implies $\beta(q'_S)<\beta(q_S)$. Writing
\[
p(\cdot\mid S;q,\alpha)=(1-\beta(q_S))p_0+\beta(q_S)q(\cdot\mid S),
\qquad
p(\cdot\mid S;q',\alpha)=(1-\beta(q'_S))p_0+\beta(q'_S)q(\cdot\mid S),
\]
shows that $p(\cdot\mid S;q',\alpha)$ lies on the line segment between $p_0$ and $q(\cdot\mid S)$
and is closer to $p_0$ because it places the smaller weight on $q(\cdot\mid S)$.

The coordinate form follows immediately: for each $x\in S$,
\[
p(x\mid S;q,\alpha)-p_0(x\mid S)
=\beta(q_S)\Big(q(x\mid S)-p_0(x\mid S)\Big).
\]
\end{proof}

\subsection{Proofs for Section~\ref{sec:identification}}
\label{app:proofs-identification}
The AWLM formula \eqref{eq:awlm-closed} can be rearranged to isolate the idiosyncratic preferences. Multiplying both sides by the normalization factor $D_i(\alpha)$ yields the \emph{pre-normalization identity}
\begin{equation}
D_i(\alpha)\,P_i(x) = (1-\alpha)p_0(x\mid S) + \alpha q_i(x),
\qquad x\in S,\ i=1,2.
\label{eq:prenorm}
\end{equation}
This identity says that the rescaled choice probability is a linear combination of the idiosyncratic distribution and the influencer's distribution, a structure we can exploit by comparing two exposures.

Subtracting \eqref{eq:prenorm} for $i=1$ from $i=2$ eliminates the idiosyncratic term, leaving only the influence:
\[
D_2(\alpha)P_2(x) - D_1(\alpha)P_1(x) = \alpha\big(q_2(x)-q_1(x)\big),
\qquad x\in S.
\]
This difference equation is the engine of identification. To express it in terms of observables, define
\begin{equation}
\Delta \equiv    P_2 - P_1 \in\mathbb{R}^S,
\qquad
A \equiv    (q_2-q_1)\big|_S + (1-q_{2,S})P_2 - (1-q_{1,S})P_1 \in\mathbb{R}^S.
\end{equation}
The vector $\Delta$ is simply the change in observed choice probabilities. The vector $A$ may look complicated, but it has a natural interpretation: it adjusts the raw difference in influencer distributions $(q_2-q_1)|_S$ for the normalization effects captured by the terms $(1-q_{i,S})P_i$. When $q_{1,S} = q_{2,S}$, $A$ can be written as $A=(q_2 - q_1)|_S + (1-q_S)\Delta$.

\begin{proof}[Proof of Proposition~\ref{prop:two-exposure-identity}]
Starting from the pre-normalization identity \eqref{eq:prenorm} and using $D_i(\alpha)=(1-\alpha)+\alpha q_{i,S}$, we expand both sides:
\begin{align*}
D_2 P_2(x) - D_1 P_1(x)
&= \big(1-\alpha+\alpha q_{2,S}\big)P_2(x) - \big(1-\alpha+\alpha q_{1,S}\big)P_1(x)\\
&= (1-\alpha)\Delta(x) + \alpha\big(q_{2,S}P_2(x)-q_{1,S}P_1(x)\big).
\end{align*}
On the other hand, subtracting \eqref{eq:prenorm} for $i=1$ from $i=2$ gives 
$D_2 P_2(x) - D_1 P_1(x) = \alpha(q_2(x)-q_1(x))$. Equating and rearranging,
\[
(1-\alpha)\Delta(x) = \alpha\Big[(q_2(x)-q_1(x)) - q_{2,S}P_2(x) + q_{1,S}P_1(x)\Big].
\]
Expanding $\alpha A(x)$ using definition \eqref{eq:A-def}:
\begin{align*}
\alpha A(x)
&= \alpha(q_2-q_1)(x) + \alpha(1-q_{2,S})P_2(x) - \alpha(1-q_{1,S})P_1(x)\\
&= \alpha(q_2-q_1)(x) + \alpha\Delta(x) - \alpha\big(q_{2,S}P_2(x)-q_{1,S}P_1(x)\big).
\end{align*}
The right-hand sides coincide, yielding $(1-\alpha)\Delta(x) = \alpha A(x) - \alpha\Delta(x)$, 
so $\Delta(x) = \alpha A(x)$ for every $x\in S$.

For the converse, suppose \eqref{eq:Delta-equals-alphaA} holds for some $\alpha\in(0,1)$. 
Define $p_0(x) \equiv [D_1(\alpha)P_1(x) - \alpha q_1(x)]/(1-\alpha)$ for $x\in S$. 
Then \eqref{eq:prenorm} holds for $i=1$ by construction. 
Substituting $P_2=P_1+\Delta$ with $\Delta=\alpha A$ and using the definition of $A$ shows 
that \eqref{eq:prenorm} also holds for $i=2$ with the same $p_0$.
\end{proof}

\begin{proof}[Proof of Proposition~\ref{prop:idiosyncratic-id}]
From \eqref{eq:prenorm}, for each exposure $q$ we have
\[
D(S,q;\alpha)p(x;S,q,\alpha) - \alpha q(x) = (1-\alpha)p_0(x\mid S),
\]
so \eqref{eq:P0-def} yields $p_0^S=p_0(\cdot\mid S)$, which is independent of $q$. 
It is immediate that $p_0^S\in\Delta(S)$. By Luce's representation theorem, 
there exists $u:\mathcal{X}\to\mathbb{R}_{++}$ such that $p_0^S(x)=u(x)/\sum_{y\in S}u(y)$; 
multiplying $u$ by any positive constant leaves $p_0^S$ unchanged, so $u$ is identified up to scale.
\end{proof}

\begin{proof}[Proof of Proposition~\ref{prop:finite-rationalizability}]
The ``only if'' direction follows by applying Proposition~\ref{prop:two-exposure-identity} 
to every pair $(i,j)$ and noting that \eqref{eq:prenorm} implies 
$b_i(\alpha)=(1-\alpha)p_0(\cdot\mid S)$ for all $i$. 

For the ``if'' direction, suppose $\Delta_{ij}=\alpha A_{ij}$ for all pairs and 
$b_i(\alpha)$ is constant across $i$. Define $p_0^S \equiv b_i(\alpha)/(1-\alpha)$; 
the nonnegativity condition ensures $p_0^S\in \Delta(S)$. Then for each $i$,
\[
D_i(\alpha)P_i = (1-\alpha)p_0^S + \alpha q_i|_S,
\]
so the AWLM representation holds on $S$ with idiosyncratic distribution $p_0(\cdot\mid S)=p_0^S$. 
Uniqueness follows from \eqref{eq:prenorm}.
\end{proof}

\subsection{Proof of Proposition~\ref{prop:generic-McManus}}
\label{app:proof-generic-id}

Fix a nonempty feasible set $S\subseteq\mathcal X$ with $|S|=m\ge 2$,
an influencer menu $I\supsetneq S$, and an exposure profile
$q=(q_1,\ldots,q_K)\in\Delta(I)^K$ with $K\ge 3$.
For each $i$, write
\[
v_i\equiv  q_i|_S\in\mathbb R^m_+,
\qquad
s_i\equiv  q_{iS}=\sum_{x\in S}q_i(x)=\mathbf 1^\top v_i\in[0,1].
\]
For parameters $(\alpha,p_0)\in(0,1)\times\Delta^{\circ}(S)$, AWLM implies that the induced choice under exposure $q_i$ satisfies, for each $x\in S$,
\[
P_i(x)=\frac{(1-\alpha)p_0(x)+\alpha v_i(x)}{D_i(\alpha)},
\qquad
D_i(\alpha)\equiv  (1-\alpha)+\alpha s_i,
\]
and we define the stacked map $\Phi_q(\alpha,p_0)\equiv  (P_1,\ldots,P_K)$.

\medskip
\noindent\text{Step 1: An open dense set of non-collinear designs.}
Let $\mathcal Q^\ast\subset\Delta(I)^K$ be the set of exposure profiles
$q=(q_1,\ldots,q_K)$ such that the restriction vectors $\{v_i\}_{i=1}^K$
do not all lie on a common affine line in $\mathbb R^m$; equivalently,
there exist indices $i,j,k$ such that $v_j-v_i$ and $v_k-v_i$ are linearly independent.
This condition can be witnessed by a strict non-vanishing of some $2\times 2$ minor,
hence $\mathcal Q^\ast$ is open in the relative topology of $\Delta(I)^K$.

To see density, fix any $q\in\Delta(I)^K$ and any neighborhood of $q$.
Approximate $q$ by an interior profile (all coordinates strictly positive), which is possible
since the relative interior is dense. If the approximating profile lies in $\mathcal Q^\ast$
we are done. Otherwise, all $\{v_i\}$ lie on a common affine line.
Pick indices $1,2$ with $v_2\neq v_1$ if possible (if all $v_i$ are equal,
any small perturbation works). Let $d\equiv  v_2-v_1\neq 0$ denote the line direction.
Choose an alternative $x\in S$ such that $e_x$ is not parallel to $d$
(this is always possible because $m\ge 2$).
Let $z\in I\setminus S$ be any infeasible alternative (nonempty by $I\supsetneq S$).
For sufficiently small $\varepsilon>0$, define a perturbed exposure $q_3'\in\Delta(I)$ by
\[
q_3'(x)\equiv  q_3(x)+\varepsilon,\qquad
q_3'(z)\equiv  q_3(z)-\varepsilon,
\qquad
q_3'(y)\equiv  q_3(y)\ \text{for all other }y\in I.
\]
Because $q_3$ is interior, choosing $\varepsilon$ small ensures $q_3'\in\Delta(I)$.
This perturbation changes $v_3$ by $v_3'\equiv  v_3+\varepsilon e_x$,
which moves $v_3$ in a direction not parallel to $d$, hence breaks affine collinearity.
Keeping all other $q_i$ unchanged yields an exposure profile $q'$ arbitrarily close to $q$
with $q'\in\mathcal Q^\ast$. Therefore $\mathcal Q^\ast$ is dense.

\medskip
\noindent\text{Step 2: Designs in $\mathcal Q^\ast$ rule out collapsed induced choices.}
We show that if there exists $(\alpha,p_0)$ such that
\[
P_1=\cdots=P_K=:P\in\Delta(S),
\]
then the restriction vectors $\{v_i\}$ must lie on a common affine line.

Using the pre-normalization identity $D_i(\alpha)P=(1-\alpha)p_0+\alpha v_i$,
take differences across $i$ and $j$:
\[
\bigl(D_i(\alpha)-D_j(\alpha)\bigr)P=\alpha(v_i-v_j).
\]
Since $D_i(\alpha)-D_j(\alpha)=\alpha(s_i-s_j)$, cancel $\alpha>0$ to obtain
\[
v_i-v_j=(s_i-s_j)P\qquad\text{for all }i,j.
\]
Thus $v_i=v_1+(s_i-s_1)P$ for all $i$, i.e., $\{v_i\}$ lie on the affine line
$\{v_1+tP:t\in\mathbb R\}$. Consequently, if $q\in\mathcal Q^\ast$,
there is no parameter value $(\alpha,p_0)$ for which $P_1=\cdots=P_K$.

\medskip
\noindent\text{Step 3: Full-rank Jacobian for $q\in\mathcal Q^\ast$.}
Fix $q\in\mathcal Q^\ast$. We show that the Jacobian of $\Phi_q$ with respect to
$(\alpha,p_0)$ has full column rank $m$ at every $(\alpha,p_0)\in(0,1)\times\Delta^{\circ}(S)$.

Let $N_i\equiv  (1-\alpha)p_0+\alpha v_i$, so $P_i=N_i/D_i(\alpha)$.
The tangent space to $\Delta^{\circ}(S)$ at $p_0$ is
\[
T_{p_0}\Delta^{\circ}(S)=\{\delta\in\mathbb R^m:\mathbf 1^\top \delta=0\}.
\]
For any $\delta\in T_{p_0}\Delta^{\circ}(S)$,
\[
D_{p_0}P_i[\delta]=\frac{1-\alpha}{D_i(\alpha)}\,\delta.
\]
Hence the $p_0$-directions contribute an $(m-1)$-dimensional subspace in the stacked output.

Differentiating with respect to $\alpha$ gives (quotient rule, using $N_i=D_iP_i$):
\[
\frac{\partial P_i}{\partial\alpha}
=\frac{v_i-p_0+(1-s_i)P_i}{D_i(\alpha)}.
\]
If the Jacobian failed to have full rank $m$ at some $(\alpha,p_0)$,
then (since the $p_0$-directions already contribute rank $m-1$)
the $\alpha$-direction would lie in their span. Thus there would exist
$\delta\in T_{p_0}\Delta^{\circ}(S)$ such that for every $i$,
\[
\frac{\partial P_i}{\partial\alpha}
=\frac{1-\alpha}{D_i(\alpha)}\,\delta.
\]
Multiplying by $D_i(\alpha)$ yields
\[
v_i-p_0+(1-s_i)P_i=(1-\alpha)\delta\qquad\text{for all }i,
\]
so the left-hand side must be independent of $i$. In particular, for any $i\neq j$,
\[
(v_j-v_i)+(1-s_j)P_j-(1-s_i)P_i=0.
\]
But under AWLM, Proposition~\ref{prop:two-exposure-identity} implies for every pair $(i,j)$,
\[
\Delta_{ij}\equiv  P_j-P_i=\alpha A_{ij},
\qquad
A_{ij}\equiv  (v_j-v_i)+(1-s_j)P_j-(1-s_i)P_i.
\]
Therefore $A_{ij}=0$ implies $\Delta_{ij}=0$, i.e.\ $P_i=P_j$ for all pairs, hence
$P_1=\cdots=P_K$. This contradicts Step~2 for $q\in\mathcal Q^\ast$.
Thus the Jacobian has full column rank $m$ everywhere.

\medskip
\noindent\text{Step 4: Global injectivity for $q\in\mathcal Q^\ast$.}
Fix $q\in\mathcal Q^\ast$ and suppose
\[
\Phi_q(\alpha,p_0)=\Phi_q(\alpha',p_0')=(P_1,\ldots,P_K).
\]
By Step~2, the induced choices cannot all coincide across exposures,
so there exist $i\neq j$ with $P_i\neq P_j$, hence $\Delta_{ij}\neq 0$.
By Proposition~\ref{prop:two-exposure-identity}, we have simultaneously
\[
\Delta_{ij}=\alpha A_{ij}
\quad\text{and}\quad
\Delta_{ij}=\alpha' A_{ij},
\]
so $\alpha=\alpha'$. With $\alpha$ fixed, the pre-normalization identity gives
\[
(1-\alpha)p_0=D_i(\alpha)P_i-\alpha v_i
=(1-\alpha)p_0'
\]
for any $i$, hence $p_0=p_0'$. This proves injectivity.

\medskip
\noindent\text{Conclusion.}
Steps~1--4 show that $\mathcal Q^\ast$ is open and dense and that for every $q\in\mathcal Q^\ast$,
the map $\Phi_q$ is injective and has full column-rank Jacobian everywhere.
This proves Proposition~\ref{prop:generic-McManus}. \qed

\section{Econometric Details (GMM derivation, Consistency)}
\label{app:gmm}

This appendix develops a two-step GMM estimator for the AWLM parameters on a fixed feasible set
$S\subseteq\mathcal X$ with $|S|=m\ge 2$. The moment conditions are based on the pre-normalization identity
\eqref{eq:prenorm}.

\subsection{Data and moment conditions}

Fix $S$ and suppose we observe $K$ exposure conditions $\{q_i\}_{i=1}^K$, where each $q_i$ is a probability
distribution on $\mathcal X$. The estimation only uses $q_{i,S}\equiv  \sum_{x\in S}q_i(x)$ and the restriction
$q_i|_S\in\mathbb R^S$.

Under exposure $q_i$, we observe $N_i$ independent choices from $S$. Let $\hat P_i\in\Delta(S)$ denote the
empirical choice shares (multinomial frequencies), with $\mathbb E[\hat P_i]=P_i$, where $P_i$ is the true
choice distribution induced by AWLM under $q_i$.

For $\alpha\in(0,1)$ define the normalization factor
\[
D_i(\alpha)\equiv   (1-\alpha)+\alpha q_{i,S}.
\]
Let $p_0\in\Delta(S)$ denote the (menu-$S$) idiosyncratic distribution. Motivated by \eqref{eq:prenorm}, define the
$ m$-dimensional sample moment vector
\[
m_i(\alpha,p_0)\equiv   D_i(\alpha)\hat P_i-(1-\alpha)p_0-\alpha\,q_i|_S\in\mathbb R^{m}.
\]
Under AWLM, the moment condition satisfies $\mathbb E[m_i(\alpha_0,p_0)]=0$ for each $i$.

\subsection{Handling the adding-up singularity via dropping one alternative}

Because $\hat P_i$ is multinomial, $\mathrm{Var}(\hat P_i)$ is singular due to the adding-up constraint.
Let $H$ be the $(m-1)\times m$ selector matrix that drops the last alternative (after indexing $S=\{1,\dots,m\}$),
\[
H\equiv   \begin{pmatrix} I_{m-1} & 0_{(m-1)\times 1}\end{pmatrix}.
\]
Define the reduced moments and reduced idiosyncratic parameter
\[
\tilde m_i(\alpha,\tilde p_0)\equiv   H\,m_i(\alpha,p_0)\in\mathbb R^{m-1},
\qquad
\tilde p_0\equiv   H p_0\in\mathbb R^{m-1}.
\]
Equivalently,
\[
\tilde m_i(\alpha,\tilde p_0)=H\!\left(D_i(\alpha)\hat P_i-\alpha q_i|_S\right)-(1-\alpha)\tilde p_0.
\]
Stack $\tilde m(\theta)\equiv   (\tilde m_1(\theta)^\top,\dots,\tilde m_K(\theta)^\top)^\top\in\mathbb R^{K(m-1)}$
where $\theta\equiv   (\alpha,\tilde p_0)\in(0,1)\times\mathbb R^{m-1}$.

For multinomial sampling,
\[
\mathrm{Var}(\hat P_i)=\frac{1}{N_i}\Sigma(P_i),
\qquad
\Sigma(P_i)\equiv   \diag(P_i)-P_iP_i^\top,
\]
so
\[
\mathrm{Var}(\tilde m_i(\theta_0))=\frac{D_i(\alpha_0)^2}{N_i}\,\tilde\Sigma_i,
\qquad
\tilde\Sigma_i\equiv   H\Sigma(P_i)H^\top.
\]
If $P_i(x)\in(0,1)$ for all $x\in S$, then $\tilde\Sigma_i$ is positive definite. In finite samples, if some
$\hat P_i(x)=0$, one may add a small ridge to $\tilde\Sigma(\hat P_i)$ before inversion.

\subsection{Scaling and the optimal weighting matrix}

Let $N\equiv   \sum_{i=1}^K N_i$ and $\pi_i\equiv   N_i/N$. Under the multinomial model,
\[
\mathrm{Var}\!\big(\sqrt N\,\tilde m_i(\theta_0)\big)=\frac{D_i(\alpha_0)^2}{\pi_i}\,\tilde\Sigma_i.
\]
Hence define
\[
\Omega\equiv   \mathrm{Var}\!\big(\sqrt N\,\tilde m(\theta_0)\big)
=\blkdiag\!\left(\frac{D_1(\alpha_0)^2}{\pi_1}\tilde\Sigma_1,\dots,\frac{D_K(\alpha_0)^2}{\pi_K}\tilde\Sigma_K\right).
\]
A feasible plug-in estimator $\hat\Omega$ replaces $(\alpha_0,P_i,\pi_i)$ with $(\tilde\alpha,\hat P_i,\hat\pi_i)$ where
$\hat\pi_i\equiv   N_i/N$, and uses $\tilde\Sigma(\hat P_i)\equiv   H\Sigma(\hat P_i)H^\top$.

\subsection{Two-step GMM and concentrating out the idiosyncratic}

For any positive definite weight matrix $W$ (typically block-diagonal), define the GMM criterion
\[
Q_N(\theta;W)\equiv   N\,\tilde m(\theta)^\top W\,\tilde m(\theta)
=\big(\sqrt N\,\tilde m(\theta)\big)^\top W\,\big(\sqrt N\,\tilde m(\theta)\big).
\]
A two-step estimator is obtained as follows:
\begin{enumerate}[label=\textup{(\roman*)}, leftmargin=2em]
\item \text{First step:} choose a convenient preliminary weight $W_0$ (e.g., $W_0=I$ or $W_0=\blkdiag(\pi_i I_{m-1})$)
and compute $\tilde\theta\equiv   \arg\min_\theta Q_N(\theta;W_0)$.
\item \text{Second step:} form $\hat\Omega$ using $(\tilde\alpha,\hat P_i,\hat\pi_i)$ and compute
\[
\hat\theta\equiv   \arg\min_\theta Q_N(\theta;\hat\Omega^{-1}).
\]
\end{enumerate}

\paragraph{Concentrating out $\tilde p_0$.}
Suppose $W$ is block-diagonal, $W=\blkdiag(W_1,\dots,W_K)$ with each $W_i\in\mathbb R^{(m-1)\times(m-1)}$.
For fixed $\alpha$, the criterion is quadratic in $\tilde p_0$ and admits the closed-form minimizer
\begin{equation}
\label{eq:tildep0-gls}
\hat{\tilde p}_0(\alpha)
=
\frac{1}{1-\alpha}\left(\sum_{i=1}^K W_i\right)^{-1}\sum_{i=1}^K W_i\,H\!\left(D_i(\alpha)\hat P_i-\alpha q_i|_S\right).
\end{equation}
This is the GLS-style weighted average. In the one-step case $W_i\propto I_{m-1}$, it reduces to the simple average
$\hat{\tilde p}_0(\alpha)\propto \sum_i H(D_i(\alpha)\hat P_i-\alpha q_i|_S)$.

One can therefore minimize the concentrated criterion
\[
Q_N^c(\alpha;W)\equiv   Q_N\big((\alpha,\hat{\tilde p}_0(\alpha));W\big)
\]
over $\alpha\in(0,1)$ (a one-dimensional search), and then recover $\hat{\tilde p}_0=\hat{\tilde p}_0(\hat\alpha)$.

Finally, recover the full idiosyncratic vector $\hat p_0\in\Delta(S)$ by setting its first $m-1$ coordinates to $\hat{\tilde p}_0$
and the last coordinate to $1-\mathbf 1^\top\hat{\tilde p}_0$ (interior idiosyncratic distribution ensures positivity asymptotically).

\subsection{Asymptotic distribution and standard errors}

Define $g(\theta)\equiv   \mathbb E[\tilde m(\theta)]$ and the Jacobian
\[
G\equiv   \left.\frac{\partial g(\theta)}{\partial\theta^\top}\right|_{\theta=\theta_0}.
\]
(Importantly, $G$ is defined from the \emph{unscaled} mean function $g(\theta)$; the $\sqrt N$ scaling is already absorbed into
$\Omega=\mathrm{Var}(\sqrt N\,\tilde m(\theta_0))$.)

Under standard regularity conditions and with the optimal weight $W=\Omega^{-1}$,
\[
\sqrt N(\hat\theta-\theta_0)\ \Rightarrow\ \mathcal N\!\left(0,\ (G^\top\Omega^{-1}G)^{-1}\right).
\]
A feasible variance estimator uses plug-in $\hat G$ and $\hat\Omega$.

\paragraph{Closed-form $G$.}
Let $P_i\equiv   \mathbb E[\hat P_i]$ denote the true choice probabilities under exposure $q_i$.
For each $i$, the block derivatives of $\mathbb E[\tilde m_i(\theta)]$ are
\[
\frac{\partial\,\mathbb E[\tilde m_i(\theta)]}{\partial \tilde p_0^\top}
=-(1-\alpha)\,I_{m-1},
\]
and (using $D_i'(\alpha)=q_{i,S}-1$)
\[
\frac{\partial\,\mathbb E[\tilde m_i(\theta)]}{\partial \alpha}
=
H\!\Big((q_{i,S}-1)P_i+p_0-q_i|_S\Big).
\]
Evaluating at $\theta_0$ and stacking $i=1,\dots,K$ yields $G\in\mathbb R^{K(m-1)\times m}$.
A feasible $\hat G$ replaces $(P_i,p_0,\alpha_0)$ with $(\hat P_i,\hat p_0,\hat\alpha)$.

\subsection{Overidentification test}

Let $d\equiv   K(m-1)$ be the number of moment conditions and $p\equiv   m$ the number of parameters in $\theta=(\alpha,\tilde p_0)$.
The model is overidentified if and only if $d>p$, i.e.,
\[
K(m-1)>m.
\]
With the two-step weight $W=\hat\Omega^{-1}$, the standard $J$ statistic is
\[
J\equiv   Q_N(\hat\theta;\hat\Omega^{-1})=N\,\tilde m(\hat\theta)^\top\hat\Omega^{-1}\tilde m(\hat\theta),
\]
and under correct specification,
\[
J\ \Rightarrow\ \chi^2_{\mathrm{df}},
\qquad
\mathrm{df}=d-p=K(m-1)-m.
\]
(When $\mathrm{df}=0$, the model is just-identified and $J\equiv   0$ asymptotically.)

\subsection{Least-squares estimation}
\label{subsec:ls-estimation}

In applications, observed choice frequencies will not satisfy the exact identities $\Delta=\alpha A$ and $b_i(\alpha)$ constant due to sampling noise and model misspecification. The two-exposure identity nevertheless suggests simple least-squares estimators.

For a single pair $(i,j)$, define $\Delta_{ij}$ and $A_{ij}$ as above. The AWLM restriction $\Delta_{ij}=\alpha A_{ij}$ is a linear regression of $\Delta_{ij}$ on $A_{ij}$ through the origin. The least-squares estimator minimizes $\|\Delta_{ij}-\alpha A_{ij}\|^2$ and has the closed form
\begin{equation}
\hat\alpha_{ij}
=
\frac{\langle \Delta_{ij}, A_{ij} \rangle}{\langle A_{ij}, A_{ij}\rangle},
\label{eq:alpha-ls-pair}
\end{equation}
whenever $A_{ij}\neq 0$, where $\langle\cdot,\cdot\rangle$ denotes the Euclidean inner product on $\mathbb{R}^S$.

With $K\ge 2$ exposures, we can pool all pairs and choose $\hat\alpha$ to minimize the aggregate discrepancy
\[
\sum_{1\le i<j\le K} \|\Delta_{ij}-\alpha A_{ij}\|^2.
\]
This again yields a closed-form estimator obtained by stacking the regressors $A_{ij}$ and responses $\Delta_{ij}$:
\begin{equation}
\hat\alpha
=
\frac{\sum_{1\le i<j\le K} \langle \Delta_{ij}, A_{ij}\rangle}
     {\sum_{1\le i<j\le K} \langle A_{ij}, A_{ij}\rangle}.
\label{eq:alpha-ls-multi}
\end{equation}
Given $\hat\alpha$, we can construct idiosyncratic distribution estimates
\[
\hat b_i(\hat\alpha)=D_i(\hat\alpha)P_i-\hat\alpha\,q_i|_S
\]
and average them across $i$ to obtain a least-squares estimate
\[
\hat p_0^S = \frac{1}{1-\hat\alpha}\cdot\frac{1}{K}\sum_{i=1}^K \hat b_i(\hat\alpha)
\]
of the idiosyncratic distribution. If desired, $\hat p_0^S$ can be projected back onto the simplex $\Delta(S)$ to enforce exact nonnegativity and normalization.

A detailed econometric analysis of these estimators is beyond the scope of this paper, but the algebraic structure $\Delta=\alpha A$ provides a straightforward route to estimation and specification testing in finite samples.

\subsection{Identification: Numerical examples}
\label{subsec:id-numerical}

This subsection collects a few simple numerical examples that illustrate how
the identities in this section recover the influence parameter $\alpha$ and
the idiosyncratic distribution $p_0(\cdot\mid S)$ from observed exposure--choice
pairs.

\begin{example}[Recovering $\alpha$ from two exposures with equal $q_S$]
\label{ex:id-equal-qS}
Let $\mathcal{X}=\{a,b,c\}$ and $S=\{a,b\}$. The decision-maker's idiosyncratic preferences are $u(a)=3$, $u(b)=1$, so the idiosyncratic Luce rule on $S$ is
\[
p_0(a\mid S)=\frac{3}{4},\qquad
p_0(b\mid S)=\frac{1}{4}.
\]
Suppose the true influence parameter is $\alpha=\tfrac{2}{5}$ and consider
two exposures:
\[
q_1 = \Big(\tfrac{1}{5},\tfrac{3}{10},\tfrac{1}{2}\Big),\qquad
q_2 = \Big(\tfrac{2}{5},\tfrac{1}{10},\tfrac{1}{2}\Big),
\]
so that $q_{1,S}=q_{2,S}=\tfrac{1}{2}$. Under AWLM the choice probabilities are
\[
P_1 = \Big(\tfrac{53}{80},\tfrac{27}{80}\Big),
\qquad
P_2 = \Big(\tfrac{61}{80},\tfrac{19}{80}\Big).
\]

From the analyst's perspective, $(q_1,P_1)$ and $(q_2,P_2)$ are observed. We
compute the differences
\[
\Delta \equiv    P_2-P_1
=
\Big(\tfrac{1}{10},-\tfrac{1}{10}\Big),
\]
and, using the definition in Proposition~\ref{prop:two-exposure-identity},
\[
A \equiv    (q_2-q_1)\big|_S + (1-q_{2,S})P_2 - (1-q_{1,S})P_1
=
\Big(\tfrac{1}{4},-\tfrac{1}{4}\Big).
\]
The two-exposure identity $\Delta = \alpha A$ then implies
\[
\alpha
=
\frac{\Delta(a)}{A(a)}
=
\frac{\tfrac{1}{10}}{\tfrac{1}{4}}
=
\frac{2}{5}
=
\frac{\Delta(b)}{A(b)}.
\]
Thus $\alpha$ is exactly recovered from a single non-degenerate pair, and in
this particular design the proportional-difference condition from
Remark~\ref{rem:qS-equal} manifests as
\[
\frac{P_2(a)-P_1(a)}{q_2(a)-q_1(a)}
=
\frac{P_2(b)-P_1(b)}{q_2(b)-q_1(b)}
=
\frac{\alpha}{(1-\alpha)+\alpha q_S}
=
\frac{\tfrac{2}{5}}{(1-\tfrac{2}{5})+\tfrac{2}{5}\cdot\tfrac{1}{2}}.
\]
\end{example}

\begin{example}[Recovering $\alpha$ and the idiosyncratic distribution when $q_S$ changes]
\label{ex:id-different-qS}
We keep the same universe $\mathcal{X}=\{a,b,c\}$, feasible set $S=\{a,b\}$,
idiosyncratic utilities $u(a)=3$, $u(b)=1$, and influence parameter
$\alpha=\tfrac{2}{5}$ as in Example~\ref{ex:id-equal-qS}. Now consider two
different exposures with \emph{different} total mass on $S$:
\[
q_1 = \Big(\tfrac{1}{5},\tfrac{3}{10},\tfrac{1}{2}\Big),
\qquad
q_2 = \Big(\tfrac{3}{10},\tfrac{2}{5},\tfrac{3}{10}\Big),
\]
so that $q_{1,S}=\tfrac{1}{2}$ and $q_{2,S}=\tfrac{7}{10}$. Under AWLM we obtain
\[
P_1 = \Big(\tfrac{53}{80},\tfrac{27}{80}\Big),
\qquad
P_2 = \Big(\tfrac{57}{88},\tfrac{31}{88}\Big).
\]

From the analyst's point of view, $(q_1,P_1)$ and $(q_2,P_2)$ are again
observed. We compute
\[
\Delta \equiv    P_2-P_1
=
\Big(-\tfrac{13}{880},\ \tfrac{13}{880}\Big),
\]
and
\[
A \equiv    (q_2-q_1)\big|_S + (1-q_{2,S})P_2 - (1-q_{1,S})P_1
=
\Big(-\tfrac{13}{352},\ \tfrac{13}{352}\Big).
\]
The ratio $\Delta(x)/A(x)$ is constant across $x\in S$ and recovers $\alpha$:
\[
\frac{\Delta(a)}{A(a)}
=
\frac{-\tfrac{13}{880}}{-\tfrac{13}{352}}
=
\frac{2}{5}
=
\frac{\Delta(b)}{A(b)}.
\]

Given $\alpha=\tfrac{2}{5}$, we can also recover the idiosyncratic distribution
$p_0(\cdot\mid S)$. For exposure $q_1$, the normalization factor is
\[
D_1(\alpha) = (1-\alpha)+\alpha q_{1,S}
= \Big(1-\tfrac{2}{5}\Big) + \tfrac{2}{5}\cdot\tfrac{1}{2}
= \frac{4}{5}.
\]
Using $b_1(\alpha)\equiv   D_1(\alpha)P_1 - \alpha q_1|_S$ from
Proposition~\ref{prop:finite-rationalizability}, we obtain
\[
b_1(\alpha)
=
\Big(\tfrac{9}{20},\tfrac{3}{20}\Big).
\]
Repeating the calculation with $(q_2,P_2)$ yields the same vector
$b_2(\alpha)=b_1(\alpha)$, as guaranteed by the theory. Dividing by
$(1-\alpha)=\tfrac{3}{5}$ gives
\[
p_0(a\mid S) = \frac{b_1(a)}{1-\alpha}
=\frac{\tfrac{9}{20}}{\tfrac{3}{5}}
=\frac{3}{4},\qquad
p_0(b\mid S) = \frac{b_1(b)}{1-\alpha}
=\frac{\tfrac{3}{20}}{\tfrac{3}{5}}
=\frac{1}{4},
\]
which exactly recovers the idiosyncratic Luce rule implied by the true utilities
$u(a)=3$ and $u(b)=1$.
\end{example}

\begin{example}[A non-identifiable pair]
\label{ex:id-degenerate}
The identity $\Delta = \alpha A$ shows when a single pair of exposures is
\emph{not} informative about $\alpha$. Consider again a fixed $S$ and two
exposures $q_1,q_2$ such that
\[
q_1\big|_S = q_2\big|_S
\quad\text{and}\quad
q_{1,S} = q_{2,S}.
\]
Under intra-aspiration irrelevance, AWLM implies $P_1 = P_2$, so
\[
\Delta = P_2-P_1 = 0.
\]
At the same time, we have $(q_2-q_1)\big|_S=0$ and $q_{1,S}=q_{2,S}$, so the
definition of $A_{12}$ yields
\[
A = (q_2-q_1)\big|_S + (1-q_{2,S})P_2 - (1-q_{1,S})P_1 = 0.
\]
Thus both $\Delta$ and $A$ are identically zero, and any value of $\alpha$
satisfies $\Delta=\alpha A$ for this pair. In other words, an exposure pair
that does not change the influencer's behavior on $S$ (and keeps $q_S$
fixed) cannot by itself identify the strength of influence $\alpha$.
Identification requires at least one non-degenerate pair with $A\neq 0$, as
in Examples~\ref{ex:id-equal-qS} and~\ref{ex:id-different-qS}.
\end{example}

\section{Additional Figures \& Intuition}
\label{app:figures}
\begin{figure}[ht!]
  \centering
  \includegraphics[width=0.75\textwidth]{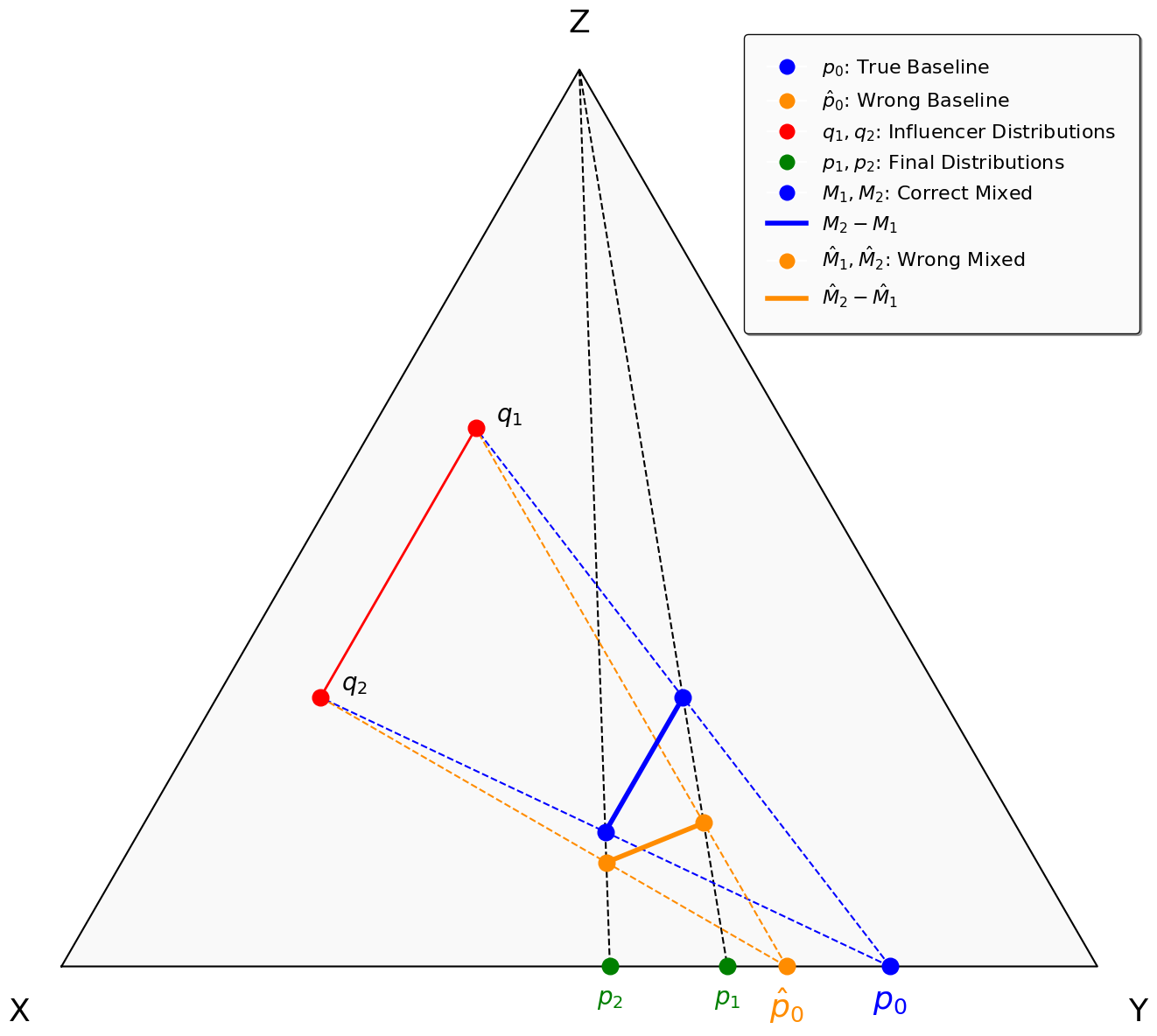}
  \caption{Analogous to Figure~\ref{fig:identification-2d}; comparison of identification geometry under correct and incorrect idiosyncratic choice distribution.}
  \label{fig:identification-3d}
\end{figure}

\begin{figure}[ht!]
  \centering
  \includegraphics[width=0.95\textwidth]{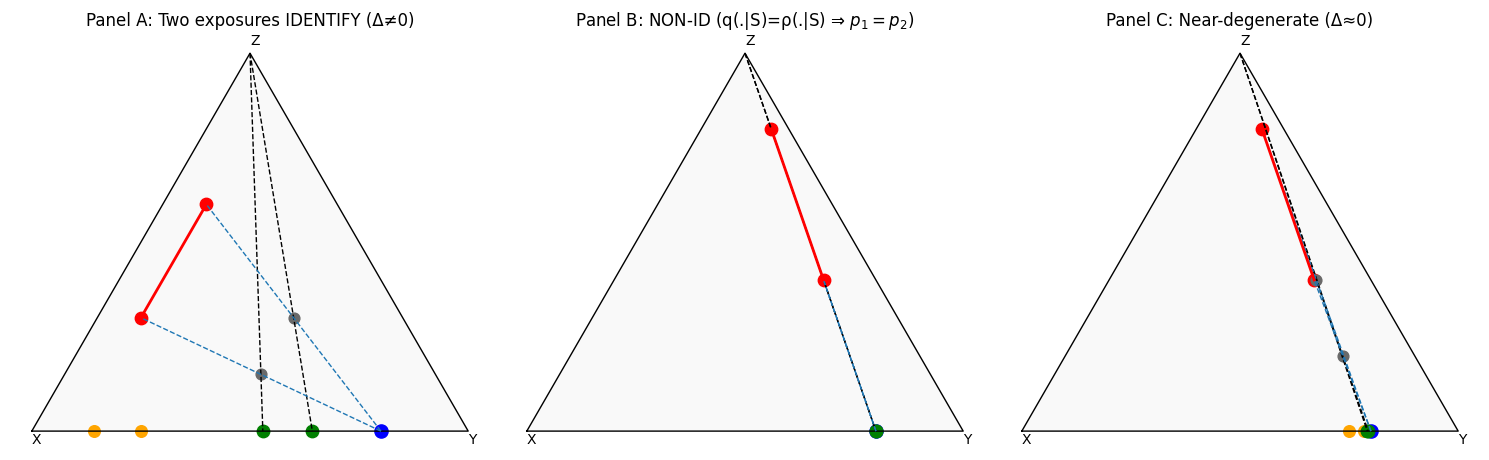}
  \caption{Identification versus non-identification. \text{Panel A:} Two exposures with sufficient variation identify $(\alpha, \rho)$; the observed choices (green points) differ, and the mixed chord is parallel to $q_2 - q_1$. \text{Panel B:} Non-identification occurs when both exposures have $q(\cdot|S) = p_0(\cdot|S)$, resulting in $p_1 = p_2$. \text{Panel C:} Near-degenerate design where exposures nearly satisfy the non-identification condition, leading to weak identification of $\alpha$.}
  \label{fig:degenerate}
\end{figure}

\begin{figure}[ht!]
  \centering
  \includegraphics[width=0.8\textwidth]{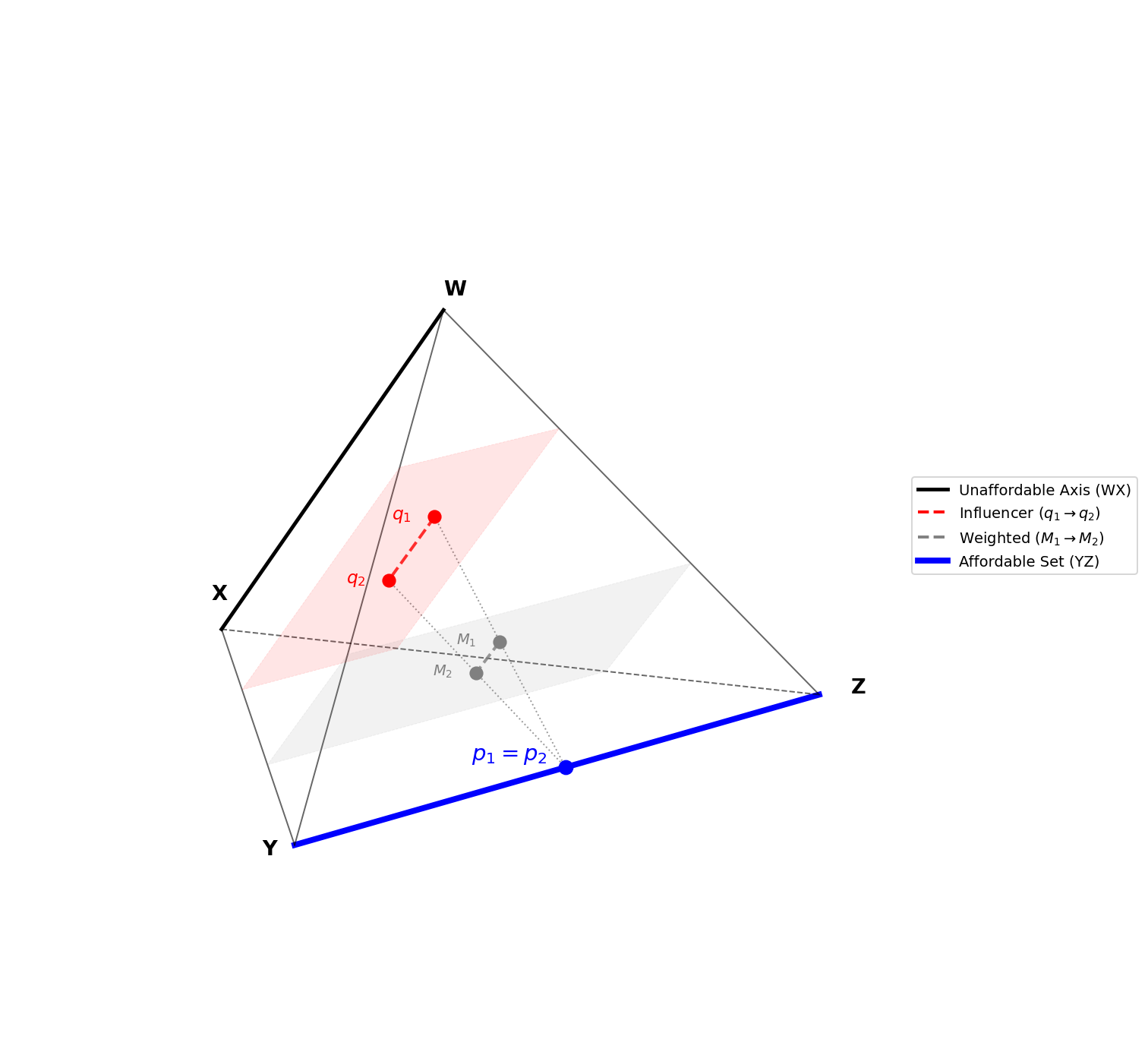}
  \caption{Intuition for Intra-Aspiration Invariance. The DM's choice on $S = \{y,z\}$ depends on the influencer's distribution $q$ only through its restriction $q|_S$. Reallocating probability mass among aspirational alternatives in $I \setminus S = \{x,w\}$ does not affect the final choice distribution on $S$, as long as $q|_S$ remains unchanged.}
  \label{fig:intra-aspiration}
\end{figure}

\begin{figure}[ht!]
  \centering
  \begin{subfigure}[b]{0.48\textwidth}
    \centering
    \includegraphics[width=\textwidth]{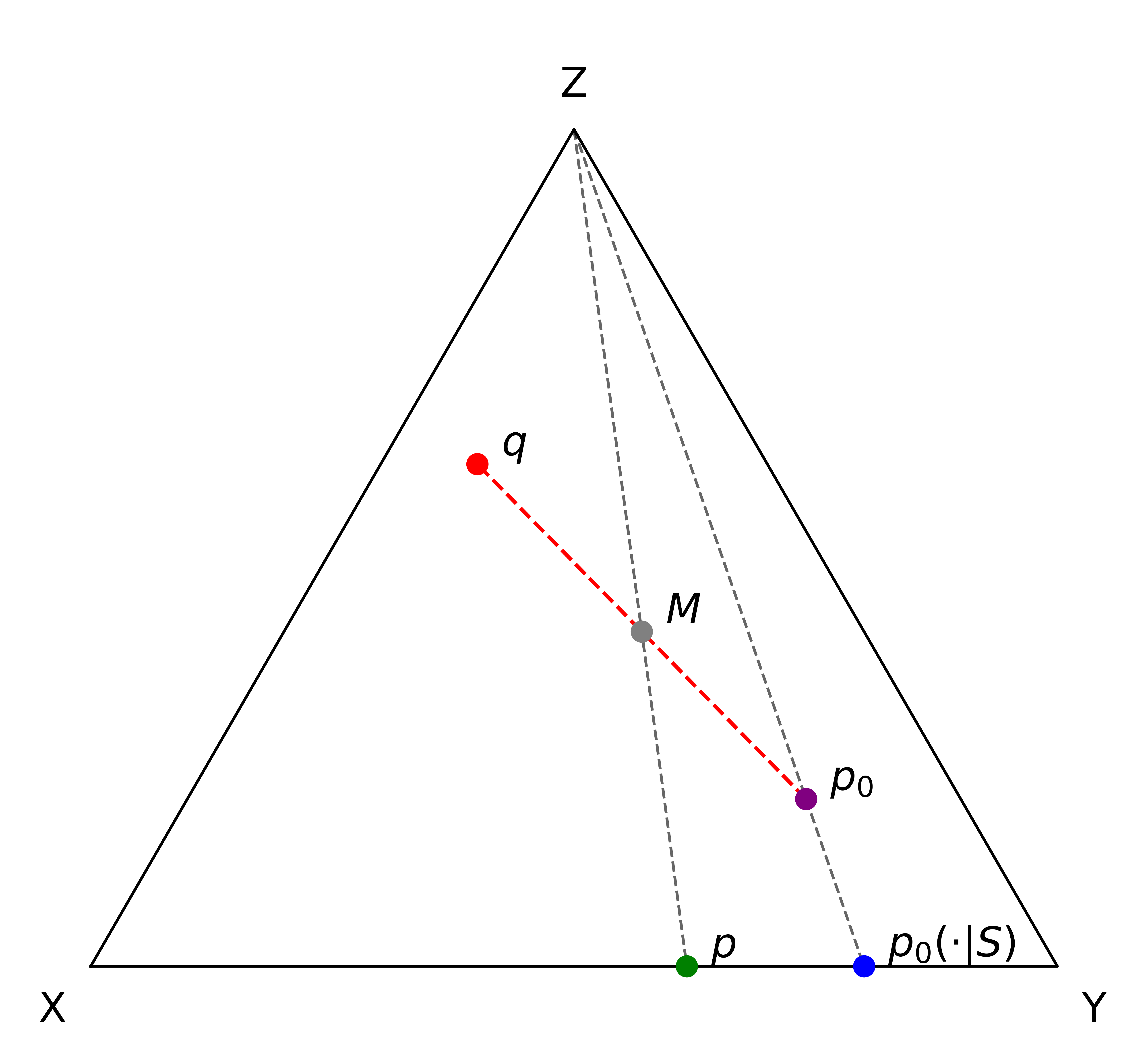}
    \caption{Mix first, then normalize}
    \label{fig:mix-first}
  \end{subfigure}
  \hfill
  \begin{subfigure}[b]{0.48\textwidth}
    \centering
    \includegraphics[width=\textwidth]{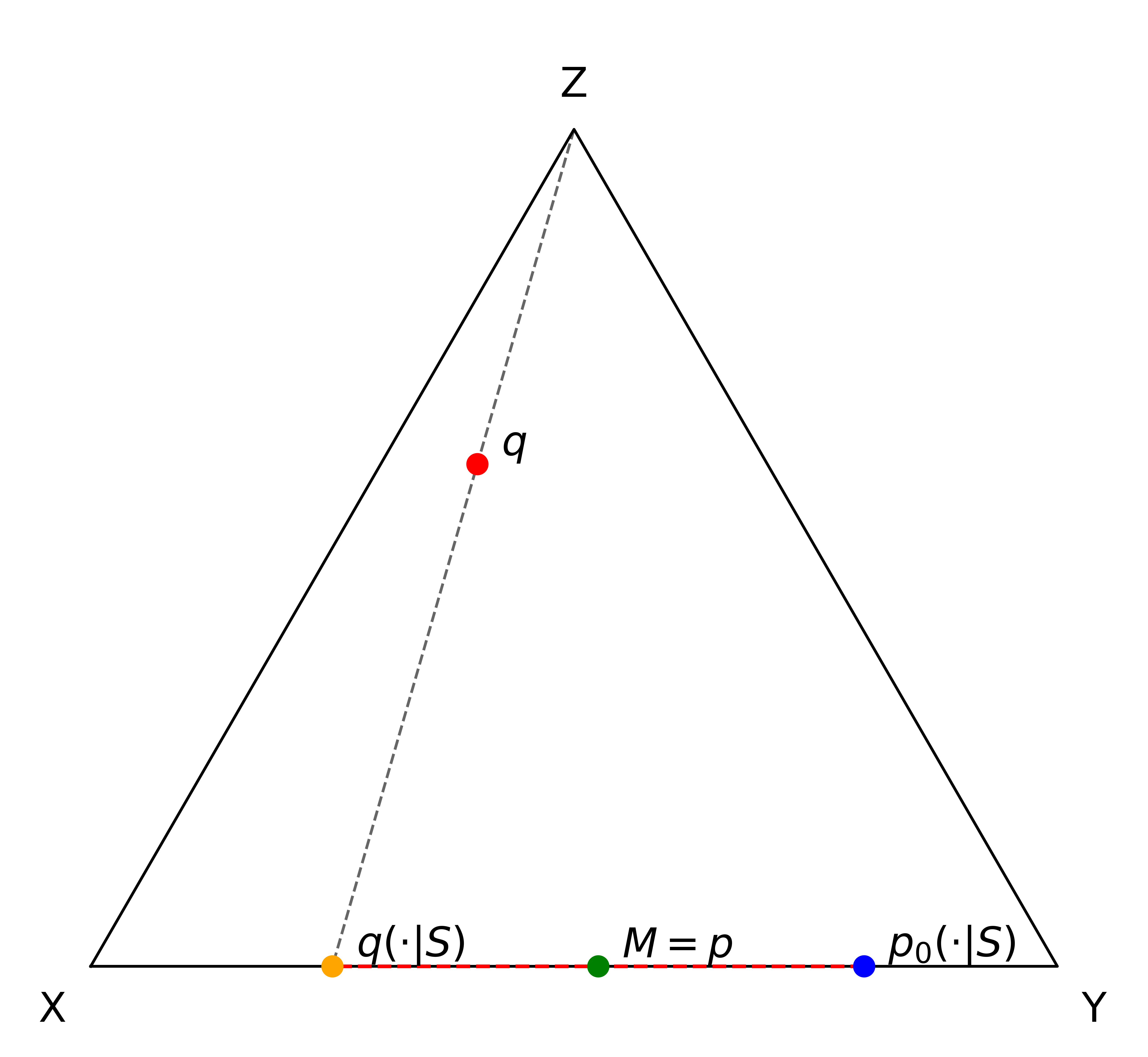}
    \caption{Project first, then mix}
    \label{fig:project-first}
  \end{subfigure}
  \caption{Alternative mixing orders. The AWLM adopts the ``mix first, then normalize'' approach (left panel): the DM's idiosyncratic distribution is first mixed with the influencer's full distribution on $I$, and then the result is normalized onto $S$. An alternative specification would first project the influencer's distribution onto $S$ and then mix (right panel). These yield different predictions when $q_S < 1$.}
  \label{fig:mixing-orders}
\end{figure}

\end{document}